\documentclass[11pt]{article}
\usepackage{a4wide}
\usepackage[usenames,dvipsnames]{color}
\usepackage{amssymb,amsmath,graphics,color,enumerate,eucal}
\usepackage{tikz}
\usepackage{xspace}
\usetikzlibrary{arrows,shapes,snakes,automata,backgrounds,petri,patterns} 
\usepackage{todonotes}
\usepackage[utf8x,utf8]{inputenc} 
\usepackage[T1]{fontenc}
\usepackage[vlined,linesnumbered,boxruled]{algorithm2e}

\usepackage{epsfig}
\usepackage{subfig}
\usepackage{amsmath,amsfonts,amssymb,epsfig,color,amsthm}
\usepackage{comment}
\usepackage{thmtools}
\usepackage{thm-restate}
\usepackage{mdframed}

\newtheorem{theorem}{Theorem}[]
\newtheorem{lemma}[theorem]{Lemma}
\newtheorem{corollary}[theorem]{Corollary}
\newtheorem{proposition}[theorem]{Proposition}

\newcommand{\retheorem}{theorem}

\newcommand{\mycase}[1]{\noindent{\bf CASE #1\ }}

\newcommand{\contractd}{\ensuremath{D_c}}%
\newcommand{\shortcutd}{\ensuremath{D_s}}%
\newcommand{\ignore}[1]{}%

\newcommand{\ProblemFormat}[1]{{\sc #1}}
\newcommand{\ProblemName}[1]{\ProblemFormat{#1}\xspace}

\newcommand{\probDS}{\ProblemName{Dominating Set}}
\newcommand{\probPlMDS}{\ProblemName{Planar Dominating Set}}

\newcommand{\probPlFVS}{\ProblemName{Planar Feedback Vertex Set}}
\newcommand{\probRMLO}{\ProblemName{$k$-Leaf Out-Branching}}
\newcommand{\probRMLOshort}{\ProblemName{LOB}}
\newcommand{\probPlRMLO}{\ProblemName{Planar Rooted Maximum Leaf Outbranching}}
\newcommand{\probHRMLO}{\ProblemName{$H$-minor-free Rooted Maximum Leaf Outbranching}}

\newcommand{\probIOB}{\ProblemName{$k$-Internal Out-Branching}}
\newcommand{\probIOBshort}{\ProblemName{IOB}}

\newcommand{\mc}{\mathcal}

\def\grad_#1{\nabla\!_{#1}}

\newcommand{\heading}[1]{\medskip\noindent{\bf #1.\ }}%

\newcounter{rulecnt}
\newcommand{\rrule}[2]{\medskip\noindent{\bf \refstepcounter{rulecnt}\label{#1}Rule~\ref{#1}}  #2\ }%

\newcommand{\Gg}{{\ensuremath{\mathcal{G}}}}
\newcommand{\Hh}{{\ensuremath{\mathcal{H}}}}

\newcommand{\nab}{\mathop{\triangledown}}

\def\ml{{\rm maxleaf}}

\def\cv{{\rm cv}}
\def\sp{{\rm sp}}
\def\iso{{\rm iso}}

\def\ea{{\rm ea}}
\def\hd{{\rm hd}}
\def\sl{{\rm sl}}

\newcommand{\scalefactor}{0.8}

\begin{document}
\title{Linear kernels for outbranching problems in sparse digraphs\thanks{Work partially supported by the ANR Grant EGOS (2012-2015) 12 JS02 002 01 (MB), by the National Science Centre of Poland, grants number 2013/09/B/ST6/03136 (\L{}K, AS). This work was done while Micha\l{} Pilipczuk has been holding a post-doc position at Warsaw Centre of Mathematics and Computer Science, and has been supported by the Foundation for Polish Science via the START stipend programme.} 
}

\date{}

\author{Marthe Bonamy\thanks{LIRMM, France} \and \L ukasz Kowalik\thanks{University of Warsaw, Poland} \and Micha\l \ Pilipczuk\footnotemark[3] \and Arkadiusz Soca\l a\footnotemark[3]}

\maketitle

\begin{abstract}
In the \probRMLO and \probIOB problems we are given a directed graph $D$ with a designated root $r$ and a nonnegative integer $k$. The question is to determine the existence of an outbranching rooted at $r$ that has at least $k$ leaves, or at least $k$ internal vertices, respectively. Both these problems were intensively studied from the points of view of parameterized complexity and kernelization, and in particular for both of them kernels with $O(k^2)$ vertices are known on general graphs. In this work we show that \probRMLO admits a kernel with $O(k)$ vertices on $\Hh$-minor-free graphs, for any fixed family of graphs $\Hh$, whereas \probIOB admits a kernel with $O(k)$ vertices on any graph class of bounded expansion.
\end{abstract}

\section{Introduction}\label{sec:intro}

Kernelization is a thriving research direction within parameterized complexity that aims at understanding the computational power of polynomial-time preprocessing procedures via a rigorous mathematical framework. Its central notion is the definition of a {\em{kernelization algorithm}}, or simply a {\em{kernel}}: Given an instance $(I,k)$ of some parameterized problem $L$, a kernelization algorithm reduces $(I,k)$ in polynomial time to an equivalent instance $(I',k')$ of $L$ so that $|I'|,k'\leq f(k)$ for some computable function $f$ of the parameter $k$ only; function $f$ is called the {\em{size}} of the kernel. While for a decidable problem $L$ the existence of {\em{any}} kernelization algorithm is equivalent to fixed-parameter tractability of the problem, we are most interested in finding small kernels, possibly of polynomial or even linear size. For concreteness, in this paper we concentrate on parameterized graph problems, so we always assume that the input instance is a graph.

One of the most influential ideas in the search for small kernels was to restrict the input graph to belong to some {\em{sparse}} graph class, e.g. to be planar, bounded-genus, or $H$-minor-free for some fixed $H$. Starting with the groundbreaking work of Alber et al.~\cite{afn:planar-domset}, who showed a kernel of size $335k$ for \probDS on planar graphs, numerous strong kernelization results were shown on planar, bounded genus, and $H$-minor-free graphs; these results often concern problems that on general graphs are intractable in the parameterized sense. A milestone in this theory is the development of the technique of {\em{meta-kernelization}} by Bodlaender et al.~\cite{BodlaenderFLPST09}, further refined by Fomin et al.~\cite{fomin:bidim-kernels}. Informally speaking, using this methodology one can explain the existence of linear kernels for many parameterized problems by proving that the problem behaves in a ``bidimensional'' way and possesses certain finite-state properties. Whereas verifying the latter usually boils down to a quick technical check, the bidimensionality requirement is quite restrictive. Roughly speaking, it says that the optimum solution size is large whenever a large two-dimensional structure (like a grid minor) can be found in the graph, and the problem behaves monotonically under minor operations.

The concept of bidimensionality was initially introduced by Demaine et al.~\cite{DemaineFHT05} as a technique for obtaining subexponential parameterized algorithms, in this case typically with the running time of the form $2^{\tilde{O}(\sqrt{k})}\cdot n^{O(1)}$. Note that, provided the considered problem can be solved in time $2^{\tilde{O}(t)}\cdot n^{O(1)}$ on graphs of treewidth $t$, the existence of a kernel with $O(k)$ vertices for the problem on planar/bounded genus/$H$-minor free graphs immediately implies an algorithm for the problem with running time $2^{\tilde{O}(\sqrt{k})}+n^{O(1)}$, as graphs from these classes have $O(\sqrt{n})$ treewidth. Thus, for many natural problems the existence of a linear kernel on sparse graphs is a stronger property than admitting a subexponential parameterized algorithm.

While the techniques of bidimensionality and meta-kernelization are elegant and have many important applications, they have certain limitations that make them inapplicable to several important families of problems, for instance problems on directed graphs or problems with prescribed sets of distinguished vertices like {\sc{Steiner Tree}}. Therefore, significant effort has been put into investigating the existence of subexponential parameterized algorithms and small kernels outside the framework of bidimensionality~\cite{DornFLRS13,FominLRS11,KleinM14,LokshtanovSW12,PilipczukPSL13,PilipczukPSL14}.

In this work we are interested in two problems investigated by Dorn et al.~\cite{DornFLRS13}, namely \probRMLO (\probRMLOshort) and \probIOB (\probIOBshort). In both problems, we are given a directed graph $D$ with a specified root $r$ and a nonnegative integer $k$. By an {\em{outbranching rooted at $r$}} we mean a spanning tree of $D$ with all the edges oriented away from~$r$. A vertex of $D$ is a {\em{leaf}} in an outbranching $T$ if it has outdegree $0$ in $T$, and is {\em{internal}} otherwise. In \probRMLOshort the question is to verify the existence of an outbranching rooted at $r$ that has at least $k$ leaves, whereas in \probIOBshort we instead ask for an outbranching rooted at $r$ with at least $k$ internal vertices. Both problems enjoy the existence of kernels with $O(k^2)$ vertices on general graphs~\cite{daligault09,Gutin-quadratic}, however up to this work no better kernels were known even in the case of planar graphs. Indeed, the directed nature of both problems prevents them from satisfying even the most basic properties needed for the bidimensionality tools to be applicable.

Dorn et al.~\cite{DornFLRS13} designed subexponential parameterized algorithms with running time $2^{\tilde{O}(\sqrt{k})}\cdot n^{O(1)}$ for both problems on $H$-minor-free graphs\footnote{We remark that Dorn et al. state the result for \probIOBshort only for apex-minor-free graphs, but a combination of their approach with the contraction decomposition technique of Demaine et al.~\cite{DemaineHK11} immediately generalizes the result to $H$-minor-free graphs.}. They did it, however, by circumventing in both cases the need of obtaining a linear kernel. In  the case of \probRMLOshort they show how to apply preprocessing rules to obtain an instance that can be still large in terms of $k$, but has treewidth $O(\sqrt{k})$ so that the dynamic programming on a tree decomposition can be applied. In the case of \probIOBshort they apply a variant of Baker's layering technique.

\paragraph*{Our results and techniques.} In this work we fill the gap left by Dorn et al.~\cite{DornFLRS13} and prove that both \probRMLOshort and \probIOBshort admit linear kernels on $H$-minor-free graphs. In fact, for \probIOBshort our approach works even in the more general setting of graph classes of bounded expansion (see Section~\ref{sec:prelims} for a definition). By slightly abusing notation, in what follows we say that a directed graph $D$ belongs to some class of undirected graphs (e.g. is $H$-minor free) if the underlying undirected graph of $D$ has this property.

\begin{restatable}{\retheorem}{restateintroRMLO}\label{thm:mainLOB}
Let $H$ be a fixed graph. There is an algorithm that, given an instance $(D,k)$ of \probRMLOshort where $D$ is $H$-minor-free, in polynomial time either resolves the instance $(D,k)$, or outputs an equivalent instance $(D',k')$ of \probRMLOshort where $|V(D')|=O(k)$, $k'\leq k$, and $D'$ is $H$-minor free.
The algorithm does not need to know $H$.
\end{restatable}

Note that Theorem~\ref{thm:mainLOB} implies also a kernel of linear size for any minor-closed family of graphs $\Gg$.
Indeed, by the Roberson and Seymour's graph minor theorem there exists a fixed finite family $\Hh$ such that $\Gg$ contains exactly graphs that are $H$-minor free for every $H\in \Hh$. 
By Theorem~\ref{thm:mainLOB}, for any input graph $D\in\Gg$, the output graph $D'$ is $H$-minor free for every $H\in \Hh$.
Hence, $D'$ is in $\Gg$.
In particular, it follows that Theorem~\ref{thm:mainLOB} implies linear kernels for planar graphs and other graphs embeddable on a surface of bounded genus.

\begin{restatable}{\retheorem}{restateintroIOB}\label{thm:mainIOB}
Let $\mc G$ be a hereditary graph class of bounded expansion. There is an algorithm that, given an instance $(D,k)$ of \probIOBshort where $D\in \mc G$, in polynomial time either resolves the instance $(D,k)$, or outputs an equivalent instance $(D',k)$ of \probIOBshort where $|V(D')|=O(k)$ and $D'$ is an induced subgraph of $D$.
\end{restatable}

By applying these kernelization algorithms and then running dynamic programming on a tree decomposition of the obtained graph, we easily obtain the following corollary.

\begin{restatable}{\retheorem}{restateintrosubexp}\label{thm:subexp}
Let $H$ be a fixed graph. Then both \probRMLOshort and \probIOBshort can be solved in time $2^{O(\sqrt{k})}+n^{O(1)}$ when the input is an $n$-vertex $H$-minor-free graph.
\end{restatable}

Algorithms with a similar running time --- but with additional $\log k$ factor in the exponent --- were obtained by Dorn et al.~\cite{DornFLRS13}. If one follows their approach, then for \probRMLOshort it is possible to shave off this factor in the exponent just by replacing the dynamic programming on a tree decomposition with a more modern one. However, for \probIOBshort the logarithmic factor is caused also by an application of the layering technique, and hence such a replacement and manipulation of parameters in layering would only improve $\log k$ to $\sqrt{\log k}$. By constructing a truly linear kernel we are able to shave this factor completely off. We remark that the running time given by Theorem~\ref{thm:subexp} is optimal under the Exponential Time Hypothesis even on planar graphs; see Section~\ref{sec:subexp} for further details.

To prove Theorems~\ref{thm:mainLOB} and~\ref{thm:mainIOB}, we revisit the quadratic kernels on general graphs given by Daligault and Thomass\'e~\cite{daligault09} (for \probRMLOshort) and by Gutin et al.~\cite{Gutin-quadratic} (for \probIOBshort). For \probRMLOshort we need to modify the approach substantially, as the core reduction rule used by Daligault and Thomass\'e is the following: whenever there is a {\em{cutvertex}} in the graph --- a vertex whose removal makes some other vertex not reachable from $r$ --- then it is safe to {\em{shortcut}} it: remove it and add an arc from every its inneighbor to every its outneighbor. Observe that an application of this rule does not preserve $H$-minor-freeness, so the kernel of Daligault and Thomass\'e~\cite{daligault09} may start with an $H$-minor free graph and go outside of this class. 

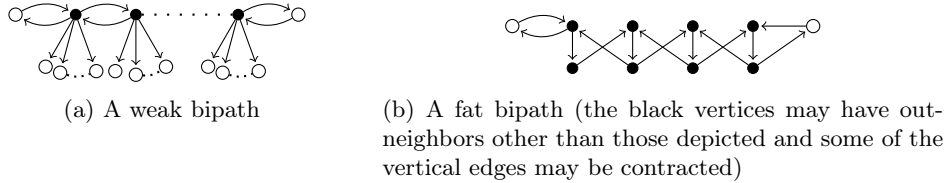
\begin{figure}[t]
\centering
%
%
%
%
%
\subfloat[][A weak bipath]{
\centering
\begin{tikzpicture}[scale=\scalefactor]
\tikzstyle{whitenode}=[draw,circle,fill=white,minimum size=5pt,inner sep=0pt]
\tikzstyle{blacknode}=[draw,circle,fill=black,minimum size=4pt,inner sep=0pt]
\tikzstyle{texte}=[circle,minimum size=5pt,inner sep=0pt]
\draw (0,0) node[whitenode] (u) {}
 ++(0:1cm) node[blacknode] (y1) {}
  ++(0:1cm) node[blacknode] (y2) {}
   ++(0:1.7cm) node[blacknode] (yp) {}
      ++(0:1cm) node[whitenode] (v) {};  
      
      \draw (u) edge [post,bend left] node {} (y1); 
      \draw (y1) edge [post,bend left] node {} (u); 
      
      \draw (y1)
      ++ (-120:1cm) node[whitenode] (y11) {};
      \draw (y1) edge [post] node {} (y11); 
      \draw (y1)
      ++ (-105:1cm) node[whitenode] (y12) {};
      \draw (y1) edge [post] node {} (y12); 
      
      \draw (y1)
      ++ (-70:1cm) node[whitenode] (y13) {};
      \draw (y1) edge [post] node {} (y13); 

      \draw[dotted, thick,bend right=20] (y12) edge node {} (y13);

      \draw (y2)
      ++ (-110:1cm) node[whitenode] (y11) {};
      \draw (y2) edge [post] node {} (y11); 
      \draw (y2)
      ++ (-90:1cm) node[whitenode] (y12) {};
      \draw (y2) edge [post] node {} (y12); 
      
      \draw (y2)
      ++ (-60:1cm) node[whitenode] (y13) {};
      \draw (y2) edge [post] node {} (y13); 

      \draw[dotted, thick,bend right=20] (y12) edge node {} (y13);

      \draw (yp)
      ++ (-120:1cm) node[whitenode] (y11) {};
      \draw (yp) edge [post] node {} (y11); 
      \draw (yp)
      ++ (-105:1cm) node[whitenode] (y12) {};
      \draw (yp) edge [post] node {} (y12); 
      
      \draw (yp)
      ++ (-70:1cm) node[whitenode] (y13) {};
      \draw (yp) edge [post] node {} (y13); 

      \draw[dotted, thick,bend right=20] (y12) edge node {} (y13);

      \draw (y2) edge [post,bend left] node {} (y1); 
      \draw (y1) edge [post,bend left] node {} (y2);
      
         \draw[loosely dotted, thick] (y2) edge node {} (yp); 
      
      \draw (yp) edge [post,bend left] node {} (v); 
      \draw (v) edge [post,bend left] node {} (yp);
\end{tikzpicture}
\label{fig:weakbipath}
}
\qquad
\subfloat[][A fat bipath (the black vertices may have out-neighbors other than those depicted and some of the vertical edges may be contracted)]{
\centering
\begin{tikzpicture}[scale=\scalefactor]
\useasboundingbox (-2cm,-0.5cm) rectangle (7cm, -1cm);
\tikzstyle{whitenode}=[draw,circle,fill=white,minimum size=5pt,inner sep=0pt]
\tikzstyle{blacknode}=[draw,circle,fill=black,minimum size=4pt,inner sep=0pt]
\tikzstyle{texte}=[circle,minimum size=5pt,inner sep=0pt]
\draw (0,0) node[whitenode] (u) {}
 ++(0:1cm) node[blacknode] (y1) {}
  ++(0:1cm) node[blacknode] (y2) {}
   ++(0:1cm) node[blacknode] (y3) {}
   ++(0:1cm) node[blacknode] (y4) {}
      ++(0:1cm) node[whitenode] (v) {};  
      
      \draw (u) edge [post,bend left] node {} (y1); 
      \draw (y1) edge [post,bend left] node {} (u); 
      
      \draw (y1) 
       ++(-90:0.7cm) node[blacknode] (y1c) {};
       \draw (y1) edge [post] node {} (y1c); 
             
\draw (y2)
++ (-90:0.7cm) node[blacknode] (y2c) {};     
\draw (y3)
++ (-90:0.7cm) node[blacknode] (y3c) {};  
\draw (y4)
++ (-90:0.7cm) node[blacknode] (y4c) {};

\draw (y2) edge [post] node {} (y2c);
\draw (y3) edge [post] node {} (y3c);
\draw (y4) edge [post] node {} (y4c);

\draw (y2c) edge [post] node {} (y1);
\draw (y1c) edge [post] node {} (y2);

\draw (y3c) edge [post] node {} (y2);
\draw (y2c) edge [post] node {} (y3);

\draw (y4c) edge [post] node {} (y3);
\draw (y3c) edge [post] node {} (y4);

\draw (y4) edge [pre] node {} (v);
\draw (y4c) edge [post] node {} (v);

\end{tikzpicture}
\label{fig:fatbipath}
}
\caption{Different types of bipaths
}
\label{fig:bipaths}
\vspace{-5mm}
\end{figure}

To circumvent this problem, we exploit the structural approach proposed by Dorn et al.~\cite{DornFLRS13}. While not achieving a linear kernel in the precise sense, Dorn et al. are able to simplify the structure of the instance so that it fits their purposes. The main idea is to contract cutedges instead of shortcutting cutvertices, which is a weaker operation that, however, preserves $H$-minor-freeness. Dorn et al. are able to expose a set of so-called {\em special} vertices $S$ of size linear in $k$ such that $G\setminus S$ has constant pathwidth; this is already enough to employ the bidimensionality technique. To obtain a linear kernel, we need to perform a much more refined analysis of the instance. More precisely, we construct a set $S$ with $|S|=O(k)$ such that $G\setminus S$ is consists of {\em{fat bipaths}}: chains as depicted in Figure~\ref{fig:bipaths}, possibly with some vertical (cut)edges contracted, and with outgoing edges with heads in $S$. After contracting the vertical edges, such a fat bipath becomes a weak bipath: a bidirectional path possibly with outgoing edges with heads in $S$. Weak bipaths are crucial in the structural approach of Daligault and Thomass\'e~\cite{daligault09}, and our fat bipaths can be thought of as more fuzzy variants of weak bipaths that cannot be reduced due to the inability to shortcut cutvertices.

To obtain a linear kernel, we need to reduce the total length of the fat bipaths. For this, we use concepts borrowed from the analysis of graph classes of bounded expansion, of which $H$-minor-free classes are special cases. Very recently Drange et al.~\cite{DrangeDFKLPPRSVS14} announced a linear kernel for \probDS on graph classes of bounded expansion, and the main tool used there is the analysis of the number of different neighborhoods that can arise in a graph $G$ from a bounded expansion graph class $\mc G$. Essentially, there is a constant $c$ such that for every $X\subseteq V(G)$ there are only $O(|X|)$ vertices in $V(G)\setminus X$ that neighbor more than $c$ vertices in $X$, while the vertices of $V(G)\setminus X$ that neighbor at most $c$ vertices in $X$ can be grouped into $O(|X|)$ classes with exactly the same neighborhoods. We apply this idea to the instance at hand with the interior of every fat bipath contracted to one vertex. Thus, we infer that there are only $O(k)$ fat bipaths that neighbor more than $c$ special vertices, and their total length can be bounded by $O(k)$ using reduction rules. On the other hand, fat bipaths with neighborhoods of size at most $c$ are reduced within their neighborhood classes, whose number is also $O(k)$.

The same neighborhood diversity argument plays the key role also in our kernel for \probIOBshort (Theorem~\ref{thm:mainIOB}). The idea of Gutin et al.~\cite{Gutin-quadratic} is that if a solution to the instance cannot be found immediately by a simple local search, then one can expose a vertex cover $U$ of size at most $2k$ in the graph. The vertices of $V(D)\setminus U$ are reduced using an argument involving crown decompositions in an auxiliary graph where vertices of $V(D)\setminus U$ are matched to pairs of adjacent vertices of $U$; this gives a quadratic dependence on $k$ of the size of the kernel. We observe that in case $D$ belongs to a class of bounded expansion, then there is only $O(|U|)=O(k)$ vertices of $V(D)\setminus U$ that have super-constant neighborhood size in $U$, while the others are grouped into $O(|U|)=O(k)$ neighborhood classes, each of which can be reduced to constant size using the same approach via crown decompositions.

For \probIOBshort we did not need any edge contractions in the reduction rules, so the kernelization procedure works on any graph class of bounded expansion. However, for \probRMLOshort it seems necessary to apply contractions of subgraphs of unbounded diameter, e.g. to reduce long paths that contribute with at most one leaf to the solution. While the last phase relies mostly on the bounded expansion properties of the graph class, we need to allow contractions in the reduction rules and hence we do not achieve the same level of generality as for \probIOBshort. 

We see the additional advantage of our approach in its simplicity. Instead of relying on complicated decomposition theorems for $H$-minor free graphs, which is a standard technique in such a setting, we use the methodology proposed by Drange et al.~\cite{DrangeDFKLPPRSVS14}: To exploit purely combinatorial, abstract notions of sparsity, like the concept of bounded expansion, and in this manner obtain a much cleaner treatment of the considered graph classes. Of particular interest is the usefulness of the approach of grouping vertices according to their neighborhoods in some fixed modulator $X$, which is the key idea in~\cite{DrangeDFKLPPRSVS14}. 

\paragraph*{Organization of the paper.} In Section~\ref{sec:prelims} we give preliminaries on tools borrowed from the analysis of graph classes of bounded expansion. Sections~\ref{sec:LOB-kernel} and~\ref{sec:IOB-kernel} are devoted to the proofs of Theorems~\ref{thm:mainLOB} and~\ref{thm:mainIOB}, respectively. In Section~\ref{sec:subexp} we derive Theorem~\ref{thm:subexp} as a corollary, and discuss the optimality of the obtained algorithms. We conclude with some closing remarks in Section~\ref{sec:conclude}.

\paragraph*{Notation.}
In this paper we deal with digraphs. 
Let $D=(V,E)$ be a digraph.
Consider an edge $(u,v)\in E$. We say that $v$ is an {\em out-neighbor} of $u$ and $u$ is an {\em in-neighbor} of $v$.
We also say that $v$ is a {\em head} and $u$ is a {\em tail} of $(u,v)$.
Also, $v$ and $u$ are {\em neighbors} of each other.
For any vertex $v$ we denote the sets of all its neighbors, out-neighbors and in-neighbors by $N_D(v)$, $N_D^+(v)$ and $N_D^-(v)$, respectively.
Moreover, the {\em degree}, {\em out-degree}, and {\em in-degree} of $v$ are defined as $\deg_D(v)=|N(v)|$, $\deg^+_D(v)=|N^+(v)|$, and $\deg^-_D(v)=|N^-(v)|$.
We omit the subscripts and write simple $N(v)$ or $\deg(v)$ whenever it does not lead to ambiguity.
For any set $S\subseteq V$ we denote $N_D^-(S)=\bigcup_{v\in S}N_D^-(v)\setminus S$ and $N_D^+(S)=\bigcup_{v\in S}N_D^+(v)\setminus S$.

\section{Preliminaries on Sparse Graphs}\label{sec:prelims}

In this section we recall some definitions and basic properties of sparse graphs, in particular $d$-degenerate graphs, bounded expansion graphs and $H$-minor-free graphs.
Although in this section we refer to undirected graphs, all the notions and claims apply also to digraphs, by looking at the underlying undirected graph.

We say that graph $G$ is $k$-degenerate when every subgraph of $G$ has a vertex of degree at most $k$. 
This implies (and in fact is equivalent to) that we can remove all the edges of $G$ by repeatedly removing vertices of degree at most $k$.
It follows that $G$ has at most $k|V(G)|$ edges.
The {\em degeneracy} of a graph is the smallest value of $k$ for which it is $k$-degenerate. 
Degeneracy is closely linked to {\em arboricity}, i.e., minimum number ${\rm arb}(G)$ of forests that cover the edges of $G$: it is well known that degeneracy is between ${\rm arb}(G)$ and $2\,{\rm arb}(G)$.

Recall that a graph $H$ is a {\em{minor}} of graph $G$ if there exists a {\em{minor model}} $(I_u)_{u\in V(H)}$ of $H$ in $G$ that satisfies the following properties:
\begin{itemize}
\item sets $I_u$ for $u\in V(H)$ are pairwise disjoint subsets of $V(G)$ that moreover induce connected subgraphs;
\item for each $uv\in E(H)$, there exist $x_u\in I_u$ and $x_v\in I_v$ such that $x_ux_v\in E(G)$.
\end{itemize}
For any fixed graph $H$, the class of {\em{$H$-minor-free graphs}} comprises all the graphs $G$ that do not have $H$ as a minor. 
Note that $H$-minor free graphs are closed under minor operations: vertex and edge deletions, and edge contractions.
For example, graphs embeddable into a constant genus surface are $H$-minor-free for some fixed $H$; in particular, by Kuratowski's theorem, planar graphs are $K_5$-minor free and $K_{3,3}$-minor free. 
The celebrated Decomposition Theorem of Robertson and Seymour implies, in a sense, a reverse implication: every $H$-minor-free graph can be decomposed into parts that are close to being embeddable into surfaces of bounded genus. The following lemma provides a connection between $H$-minor-free graphs and degeneracy.

\begin{lemma}[see Lemma 4.1 in~\cite{Sparsity}]
\label{lem:degeneracy-h-minor}
 Any $H$-minor free graph is $d_H$-degenerate for $d_H=O(|H|\sqrt{\log |H|})$.
\end{lemma}

Let $r$ be a nonnegative integer. If a minor model $(I_u)_{u\in V(H)}$ satisfies in addition that $G[I_u]$ has radius at most $r$ for each $u\in V(H)$, then $(I_u)_{u\in V(H)}$ is an {\em{$r$-shallow minor model of $H$}}, and we say that $H$ is an {\em{$r$-shallow minor of $G$}}. If $\mc G$ is a class of graphs, then by $\mc G \nab r$ we denote the class of all $r$-shallow minors of graphs from $\mc G$; note that $\mc G\nab 0$ are all subgraphs of graphs of $\mc G$. We now define the {\em{greatest reduced average degree}} ({\em{grad}}) of a class $\mc G$ at depth $r$ as 
$$\grad_r(\mc G)=\sup_{H\in \mc G\nab r} \frac{|E(H)|}{|V(H)|}.$$
That is, we take the greatest edge density among the $r$-shallow minors of $\mc G$. Class $\mc G$ is said to be of {\em{bounded expansion}} if $\grad_r(\mc G)$ is a finite constant for every $r$. Observe that then the graphs from $\mc G$ are in particular $d$-degenerate for $d=\lfloor 2\grad_0(\mc G)\rfloor$. For a single graph $G$, we denote $\grad_r(G)=\grad_r(\{G\})$.

Consider the class $\mc G_H$ of $H$-minor-free graphs. 
By Lemma~\ref{lem:degeneracy-h-minor}, every graph $G\in \mc G_H$ has at most $d_H\cdot |V(G)|$ edges. Since $\mc G_H$ is closed under taking minors, it follows that $\mc G_H\nab r=\mc G_H$ for every nonnegative $r$, so also $\grad_r(\mc G_H)\leq d_H$. Thus, $H$-minor-free graphs form a class of bounded expansion with all the grads bounded independently of $r$.

In this paper we do not use the original definition of bounded expansion graphs, but we rather rely on the point of view of diversity of neighborhoods, which was found to be very useful in~\cite{DrangeDFKLPPRSVS14}. More precisely, we now use the following result from \cite[Lemma~6.6]{BndExpKernels13}; the statement with adjusted notation is taken verbatim from~\cite{DrangeDFKLPPRSVS14}.

\begin{proposition}[Proposition 2.5 of \cite{DrangeDFKLPPRSVS14}]\label{prop:bipartite-general}
  Let $G$ be a graph, $X \subseteq V(G)$ be a vertex subset, and $R=V(G)\setminus X$.
  Then for every integer $p \geq \grad_1(G)$ it holds that
  \begin{enumerate}
  \item $| \{ v \in R \colon |N(v)\cap X| \geq 2p \}| \leq 2p \cdot
    |X|$, and
  \item $| \{ A \subseteq X \colon |A| < 2p \textrm{ and } \exists_{v\in R}\ A =
      N(v) \cap X \}| \leq (4^p + 2p) |X|$.
  \end{enumerate}
  Consequently, the following bound holds:
  \[
   | \{ A\subseteq X \colon \exists_{v\in R}\ A = N(v)\cap X \}| \leq 
  \left( 4^{\grad_1(G)} + 4\grad_1(G) \right) \cdot |X| .
  \]
\end{proposition}

We need a strengthening of the first claim of Proposition~\ref{prop:bipartite-general}.

\begin{lemma}
\label{lem:big-degrees}
 Let $G=(X,Y,E)$ be a bipartite graph of degeneracy at most $d$. Then,
 \[\sum_{\substack{y\in Y\\ \deg_G(y) > 2d}}\deg_G(y) \leq 2d|X|.\]
\end{lemma}
\begin{proof}
Let $Z=\{y\in Y \colon \deg_G(y) > 2d\}$. 
Consider $G'=G[X\cup Z]$, and observe that $|E(G')| = \sum_{y\in Z} \deg_G(y)$. Since $G'$ is a subgraph of $G$, we obtain that
$$d \geq \frac{|E(G')|}{|V(G')|} = \frac{\sum_{y\in Z} \deg_G(y)}{|X|+|Z|}.$$
Consequently $\sum_{y\in Z} \deg(y)\leq d|X|+d|Z|$, so 
$$\sum_{y\in Z} (\deg_G(y) - d)\leq d|X|.$$
Observe that since $\deg_G(y)>2d$ for each $y\in Z$, we have that $\deg_G(y)-d>\deg_G(y)/2$. 
Hence it follows that $\sum_{y\in Z} \deg_G(y) \leq 2d|X|$.
\end{proof}

Note that Proposition~\ref{prop:bipartite-general} has the following corollary when applied to $H$-minor-free graphs.

\begin{corollary}\label{cor:boundedsmallneighborhood}
Let $H$ be a graph. There exists $c_H=2^{O(|H|\sqrt{\log |H|})}$ such that in any $H$-minor-free bipartite graph $G=(X,Y,E)$, there are at most $c_H \cdot |X|$ vertices in $Y$ with pairwise distinct neighborhoods in $X$.
\end{corollary}

\section{\probRMLO in $H$-minor-free graphs}\label{sec:LOB-kernel}

In this section we deal with rooted digraphs, i.e., digraphs with a vertex $r$, called {\em root}, of in-degree 0. 
In such digraphs we redefine some standard connectivity notions as follows.
Let $(D,r)$ be a rooted digraph.
We say that $D$ is {\em connected} when every vertex of $D$ is reachable from $r$.
A {\em cut-vertex} is any vertex $v\in V(D)\setminus\{r\}$ such that $D-r$ is not connected.
The set of all cut-vertices of $D$ is denoted by $\cv(D)$.
We say that $D$ is {\em 2-connected} if $D$ has no cut-vertex (equivalently, for every vertex $v\in V(D)\setminus\{r\}$ there are at least two paths from $r$ to $v$ that do not share internal vertices).
Similarly, a {\em cut-edge} is any edge $(u,v)\in E(D)$ such that $D-(u,v)$ is not connected. 
We say that $D$ is {\em 2-edge-connected} if $D$ has no cut-edge (equivalently, for every vertex $v\in V(D)\setminus\{r\}$ there are at least two edge-disjoint paths from $r$ to $v$).
Note that if $(u,v)$ is a cut-edge then $u$ is a cut-vertex or $u=r$.


Given a cut-vertex $u$, or $u=r$, we define $P(u)$ as the set of {\em private neighbors} of $u$, that is, the set of out-neighbors of $u$ that are not reachable from the root in $D-u$. In particular, all the outneighbors of $r$ are its private neighbors.

By a {\em{contraction}} of edge $(a,b)$ in $D$ we mean the following operation: identify $a$ and $b$ into a newly introduced vertex $v_{(a,b)}$, replace $a$ and $b$ with $v_{(a,b)}$ in every edge of $D$, and remove all the loops and parallel edges created in this manner. Note that if $D$ is $H$-minor-free, then it remains $H$-minor-free after contractions as well. 

Following~\cite{daligault09}, we say that a vertex $v$ of $D$ is {\em special} if $v$ is of in-degree at least $3$ or there is an incoming simple edge, i.e., an edge $(u,v)$ such that $(v,u)\not\in E(D)$.
The set of all special vertices of $D$ is denoted by $\sp(D)$.

A {\em weak bipath} $P$ is a sequence of vertices $u_1,\ldots,u_p$ for some $p\ge 3$, such that for each $i=2,\ldots,p-1$, we have $N^-(u_i)=\{u_{i-1},u_{i+1}\} \subseteq N^+(u_i)$. 
The {\em length} of $P$ is $p-1$.
If additionally $N^+(u_i)=N^-(u_i)=\{u_{i-1},u_{i+1}\}$ for every $i=2,\ldots,p-1$, we say that $P$ is {\em proper bipath} (or shortly a {\em bipath}). $u_1$ and $u_p$ are called the {\em{extremities}} of $P$.

We say that a cut-edge $(u,v)$ is {\em lonely} when there is no other cut-edge with the tail in $u$. We call a cut-edge {\em branching} is there is another cut-edge with the same tail. 
The graph obtained from $D$ by contracting all lonely cut-edges is denoted by $\contractd$ and called the {\em contracted graph}. 
Consider a vertex $v$ of $\contractd$. Then either $v$ was created by contracting some set of cut-edges $Z$ in $D$ or $v\in D$.
In the prior case we define the {\em bag} $B$ of $v$ as the set of vertices incident to edges in $Z$.
Also, for any edge $(x,y)\in Z$ the vertex $x$ is called a {\em tail} of $B$ and $y$ is a {\em head} of $B$.
In the latter case, i.e., when $v\in D$, we define the bag as $B=\{v\}$ and $v$ is both head and tail of $B$.
When $B$ is a bag of $v$ we denote $v_B=v$ and $B_v = B$. If there is exactly one head and exactly one tail of $B$, then they are denoted by $h_B$ and $t_B$, respectively.

We say that bags $A$ and $B$ are {\em linked} if in $D$ there is an edge from $A$ to $B$ and an edge from $B$ to $A$.

\subsection{Our kernelization algorithm}\label{sec:Ok}

In this section we describe our algorithm which outputs a kernel for \probRMLO. 
The algorithm exhaustively applies reduction rules. 
Each reduction rule is a subroutine which finds in polynomial time a certain structure in the graph and replaces it by another structure, so that the resulting instance is equivalent to the original one. 
More precisely, we say that a reduction rule for parameterized graph problem $P$ is {\em correct} when for every instance $(D,k)$ of $P$ it returns an instance $(D',k')$ such that:
\begin{enumerate}[a)]
\item $(D',k')$ is an instance of $P$,
\item $(D,k)$ is a yes-instance of $P$ iff $(D',k')$ is a yes-instance of $P$, and
\item $k'\le k$.
\end{enumerate}

Below we state the rules we use. 
The rules are applied in the given order, i.e., in each rule we assume that the earlier rules do not apply.
We begin with some rules used in the previous works~\cite{daligault09}.

\begin{figure}[h]
 \centering
\begin{tikzpicture}[scale=\scalefactor]
\tikzstyle{whitenode}=[draw,circle,fill=white,minimum size=5pt,inner sep=0pt]
\tikzstyle{blacknode}=[draw,circle,fill=black,minimum size=4pt,inner sep=0pt]
\tikzstyle{texte}=[circle,minimum size=5pt,inner sep=0pt]
\draw (0,0) node[whitenode] (x) [label=90:$x$] {};  
\draw (x)
++ (180-40:1cm) node[blacknode] (y) [label=90:$y$] {};
\draw (x)
++ (180-10:2cm) node[whitenode] (z1) {};
\draw (x)
++ (180:2cm) node[whitenode] (z2) {};
\draw (x)
++ (180+20:2cm) node[whitenode] (z3) {};
       

\draw (z1) edge [post] node {} (x);
\draw (z2) edge [post] node {} (x);
\draw (z3) edge [post] node {} (x);

\draw (z1) edge [post] node {} (y);
\draw (z2) edge [post] node {} (y);
\draw (z3) edge [post] node {} (y);

\draw (y) edge [post] node {} (x);

   \draw[dotted, thick,bend right=20] (z2) edge node {} (z3);
   
   \draw (y)
++ (30:1.5cm) node[whitenode] (y1) {};
\draw (y)
++ (40:1.5cm) node[whitenode] (y2) {};
\draw (y)
++ (70:1.5cm) node[whitenode] (y3) {};

\draw (y) edge [post] node {} (y1);
\draw (y) edge [post] node {} (y2);
\draw (y) edge [post] node {} (y3);

   \draw[dotted, thick,bend right=20] (y2) edge node {} (y3);

\draw (1.5,0.2) node[texte] (k) {\Large $\leadsto$};

\draw (4.5,0) node[whitenode] (x) [label=90:$x$] {};  
\draw (x)
++ (180-40:1cm) node[blacknode] (y) [label=90:$y$] {};
\draw (x)
++ (180-10:2cm) node[whitenode] (z1) {};
\draw (x)
++ (180:2cm) node[whitenode] (z2) {};
\draw (x)
++ (180+20:2cm) node[whitenode] (z3) {};
       

\draw (z1) edge [post] node {} (x);
\draw (z2) edge [post] node {} (x);
\draw (z3) edge [post] node {} (x);

\draw (z1) edge [post] node {} (y);
\draw (z2) edge [post] node {} (y);
\draw (z3) edge [post] node {} (y);


   \draw[dotted, thick,bend right=20] (z2) edge node {} (z3);
   
   \draw (y)
++ (30:1.5cm) node[whitenode] (y1) {};
\draw (y)
++ (40:1.5cm) node[whitenode] (y2) {};
\draw (y)
++ (70:1.5cm) node[whitenode] (y3) {};

\draw (y) edge [post] node {} (y1);
\draw (y) edge [post] node {} (y2);
\draw (y) edge [post] node {} (y3);

   \draw[dotted, thick,bend right=20] (y2) edge node {} (y3);
   
\end{tikzpicture}
\caption{Rule~\ref{r:cutneighbor}}
\label{fig:r4}
\end{figure}

\rrule{r:unreach}{If there exists a vertex not reachable from $r$ in $D$, then reduce to a trivial no-instance.}

\rrule{r:cut-edge}{If there exists a cut-vertex $v$ with exactly one incoming edge $e$ then contract $e$. 
                  Similarly, if there exists a cut-vertex $v$ with exactly one outgoing edge $e$ then contract $e$.}

\rrule{r:bipath}{Let $P$ be a proper bipath of length $4$ in $D$. Contract any edge of $P$.}

\rrule{r:cutneighbor}{Let $x$ be a vertex of $D$. If there exists $y \in N^-(x)$ such that the removal of $N^-(x)\setminus \{y\}$ disconnects $y$ from $r$, then delete the edge $(y,x)$.}

\medskip

The correctness of the above reduction rules was proven in~\cite{daligault09}. 
(In~\cite{daligault09}, Rule~\ref{r:cut-edge} is formulated in a more general way, but we restrict it so that if the input digraph was $H$-minor-free, then so is the resulting reduced graph.)
Let us remark that Rule~\ref{r:cutneighbor} remains true if $r\in N^-(x)\setminus\{y\}$, and in this case it triggers removal of all the incoming edges apart from the one coming from the root.
Below we introduce two simple rules which will make our argument a bit easier.

\rrule{r:two-cut-edges}{If there are two cut-edges $(x_1,y_1)$ and $(x_2,y_2)$ such that $(x_1,x_2),(x_2,x_1)\in E(D)$, then contract $(x_1,x_2)$.}

\rrule{r:cut-double-edge}{If there is a cut-edge $(u,v)$ such that $(v,u)\in E(D)$, then remove $(v,u)$.}

\begin{lemma}\label{cl:r:two-cut-edges}
Rule~\ref{r:two-cut-edges} is correct.
\end{lemma}
\begin{proof}
  Let $D$ and $D'$ denote the graph before and after applying the reduction.
  Let $x$ be the vertex obtained by contracting $(x_1,x_2)$.
  Let $T'$ be an outbranching in $D'$. Then an outbranching of $D$ can be obtained by the following procedure: 
\begin{itemize}
\item remove $x$ and add $x_1$ and $x_2$;
\item replace the edge from the parent $p$ of $x$ by $(p,x_1)$ and $(x_1,x_2)$, or $(p,x_2)$ and $(x_2,x_1)$, depending whether $(p,x_1)\in E(D)$ or $(p,x_2)\in E(D)$;
\item for every child $c$ of $x$ if $(x_1,c)\in E(D)$, add $(x_1,c)$, otherwise add $(x_2,c)$.
\end{itemize} 
  Clearly, the number of leaves does not change.
  
  For the second direction, assume that $T$ is an outbranching of $D$.
  Then $T$ contains both $(x_1,y_1)$ and $(x_2,y_2)$, because they are cut-edges. In particular, $x_1$ and $x_2$ are not leaves in $T$.
  At least one of $x_1$, $x_2$ is not a descendant of the other in $T$, by symmetry assume $x_1$ is not a descendant of $x_2$.
  Then remove the edge from the parent of $x_2$ to $x_2$ and add the edge $(x_1,x_2)$.
  Thus we obtained an outbranching $T'$ of $D$ that contains the edge $(x_1,x_2)$ and has at least as many leaves as $T$.
  By contracting the edge $(x_1,x_2)$ in $T$ we get an outbranching of $D'$ with the same number of leaves.
\end{proof}

\begin{lemma}
Rule~\ref{r:cut-double-edge} is correct.
\end{lemma}
\begin{proof}
  Let $D$ and $D'$ denote the graph before and after applying the reduction.
  Since $D'\subseteq D$, any outbranching of $D'$ is also an outbranching of $D$.
  Pick any outbranching $T$ of $D$. Since $(u,v)$ is a cut-edge, $(u,v)\in E(T)$. Then $(v,u)\not\in E(T)$. Hence $T$ is also an outbranching of $D'$.
  It follows that $(D,k)$ is a yes-instance iff $(D',k)$ is a yes-instance.
\end{proof}

To complete the algorithm we need a final accepting rule which is applied when the resulting graph is too big. In Section~\ref{sec:analysis} we prove that Rule~\ref{r:accept} is correct for $H$-minor-free graphs for some constant $c=2^{O(|H|\sqrt{\log|H|})}$.

\rrule{r:accept}{If the graph has more than $c \cdot k$ vertices, return a trivial yes-instance (conclude that there is a rooted outbranching with at least $k$ leaves in $D$).} 
\smallskip

We conclude with the following lemma.

\begin{lemma}
 Let $H$ be a graph.
 If the input is an $H$-minor-free graph, then the output of each of the rules 1--~\ref{r:accept} is a minor of $D$, and hence an $H$-minor-free graph.
 Moreover, each rule can be recognized and applied in polynomial time, and the degree of the polynomial does not depend on $H$.
\end{lemma}

\begin{proof}
 The first claim follows from the fact that the rules modify the graph by means of deletions and contractions only.
 The second claim is straightforward to check.
\end{proof}

\subsection{A few simple properties of the reduced graph}

In this section we state simple auxiliary lemmas, which will be used in the remainder of the paper.

\begin{lemma}\label{lem:indeg}
 Assume reduction rules 1-4 do not apply to $D$.
 Let $u$ be a cut-vertex in $D$, or $u=r$.
 Then every private neighbor $v\in P(u)$ has indegree 1 and $(u,v)$ is a cut-edge. 
 In particular, the head of any cut-edge has indegree 1.
\end{lemma}

\begin{proof}
 If $v$ has indegree at least 2 then either Rule~\ref{r:unreach} applies, or Rule~\ref{r:cutneighbor} applies to $x=v$ and $y$ being the other inneighbor of $v$. 
 Any edge incoming to a vertex of indegree 1 is a cut-edge, so $(u,v)$ is a cut-edge.
 The head of any cut-edge is a private neighbor of its tail, so the last claim also follows.
\end{proof}

\begin{lemma}\label{lem:head-tail}
 If reduction rules 1-4 do not apply to $D$ then the tail of any cut-edge is not a head of another cut-edge.
\end{lemma}

\begin{proof}
 Assume $(x,y)$ and $(y,z)$ are cut-edges.
 By Lemma~\ref{lem:indeg}, $\deg_D^-(y)=1$. 
 It follows that Rule~\ref{r:cut-edge} applies, a contradiction.
\end{proof}

\begin{lemma}
\label{lem:bag-size-2}
 If reduction rules do not apply to $D$ then every bag is of size at most two and contains at most one edge. In particular, every bag has exactly one head and one tail.
\end{lemma}

\begin{proof}
Assume that there is a bag $B$ of size at least three. 
Since the cut-edges that get contracted to $v_B$ are lonely, and their heads have indegrees 1 due to Lemma~\ref{lem:indeg}, then these edges form a directed path, a contradiction with Lemma~\ref{lem:head-tail}. The fact that a bag of size~$2$ cannot contain two edges follows from Rule~\ref{r:cut-double-edge}.
\end{proof}

\begin{lemma}
\label{lem:no-edge-to-head}
 If reduction rules 1-4 do not apply to $D$, then for arbitrary pair of bags $A$ and $B$ every edge from $A$ to $B$ has head in $t_B$.
\end{lemma}

\begin{proof}
 If $|B|=1$, then the claim is trivial, so assume $|B|\ge 2$.
 By Lemma~\ref{lem:bag-size-2}, $|B|=2$, i.e., $(t_B,h_B)$ is a cut-edge.
 If there is an edge from $A$ to $B$ with head in $h_B$, then $\deg_D^-(h_B)\ge 2$, a contradiction with Lemma~\ref{lem:indeg}.
\end{proof}

\begin{lemma}\label{lem:one-edge}
Assume reduction rules 1-4 do not apply to $D$.
If bags $A$ and $B$ are linked then there is exactly one edge from $A$ to $B$ and exactly one edge from $B$ to $A$.
\end{lemma}

\begin{proof}
 It suffices to show that there is exactly one edge from $A$ to $B$, since the other claim is symmetric.
 Assume for the contradiction that there are two edges $(a_1,b_1),(a_2,b_2)\in A\times B$.
 Note that $b_1 = b_2$, for otherwise we get a contradiction with Lemma~\ref{lem:no-edge-to-head}. It follows that $a_1\ne a_2$, since there are no two identical edges in $D$.
 Assume w.l.o.g.\ that $(a_1,a_2)$ is a cut-edge. 
 Then Rule~\ref{r:cutneighbor} applies (with $x=b_1$ and $y=a_2$), a contradiction.
\end{proof}

\begin{lemma}\label{lem:root}
Assume reduction rules do not apply to $D$. Then $\deg_D^+(r)\geq 2$, all the edges going out of $r$ are branching cut-edges, and each of the outneighbors of $r$ is a special vertex in $\contractd$.
\end{lemma}
\begin{proof}
We have that $\deg_D^+(r)\geq 2$ because otherwise Rule~\ref{r:cut-edge} would apply. Therefore, it suffices to show that every edge $(r,u)$ is a cut-edge, because the head of a branching cut-edge is always special in $\contractd$ by Rule~\ref{r:cut-double-edge}. This, however, follows from inapplicability of Rule~\ref{r:cutneighbor} to $u$.
\end{proof}

\subsection{Decomposition into weak bipaths}

The following lemma gives a structural connection between weak bipaths and special vertices.
\begin{lemma}\label{lem:decomposition}
Assume no reduction rule applies to $D$. Let $S\subseteq V(\contractd)$ be any set of vertices that contains the root $r$ and every special vertex of $\contractd$. Then one can find weak bipaths $P_1,P_2,\ldots,P_q$, such that:
\begin{itemize}
\item[(i)] The sets of internal vertices of $P_1,P_2,\ldots,P_q$ form a partition of $V(\contractd)\setminus S$.
\item[(ii)] The extremities of each $P_i$ belong to $S$ and are distinct.
\item[(iii)] The out-neighbors of the internal vertices of each $P_i$ belong to $S$.
\end{itemize}
\end{lemma}
\begin{proof}
Consider any vertex $v\in V(D_c)$ such that $v\notin S$.
Assume first that $\deg_{\contractd}^-(v)=1$ and $N^-_{\contractd}(v)=\{u\}$. Since $v$ is not special in $\contractd$, we have that also $(v,u)\in E(\contractd)$. If $u,v\in D$, then $(u,v)$ would be a cut-edge in $D$ and Rule~\ref{r:cut-double-edge} would apply, a contradiction. Otherwise, the bags of $u$ and $v$ are linked and by Lemmas~\ref{lem:no-edge-to-head} and~\ref{lem:one-edge}, there is one edge from the bag of $u$ to the tail of the bag of $v$; clearly, this edge is a cut-edge in $D$. If $v$ was obtained from the contraction of a lonely cut-edge $(v_1,v_2)$, then this would be a contradiction with Lemma~\ref{lem:head-tail}. Hence assume $v\in D$. From Lemma~\ref{lem:no-edge-to-head} we infer that in $D$ there is an edge from $v$ to $t_{B_u}$. 
However, the edge from $B_u$ to $v$ has tail in $t_{B_u}$ by Lemma~\ref{lem:head-tail}. Then again Rule~\ref{r:cut-double-edge} would apply, a contradiction.

It follows that $\deg_{\contractd}^-(v)\ge 2$ for each $v\notin S$.
Since $v$ is not special, we get that $\deg_{\contractd}^-(v) = 2$, and the two of its in-neighbors are also its out-neighbors. Since $r\in S$ and $D_c$ is connected, we have that $D_c-S$ is a set of bidirectional paths, with each endpoint connected by two edges with opposite directions with a vertex of $S$. Thus we immediately obtain weak bipaths $P_1,P_2,\ldots,P_q$ that satisfy (i), (iii), as well as (ii) apart from the claim that the extremities are distinct. Suppose there is a weak bipath $P_i=u,v_2,v_3,\ldots,v_{p-1},u$ such that both its extremities are in fact one vertex $u\in S$. By Lemma~\ref{lem:root}, $u\neq r$. Regardless whether $u\in D$ or $u$ is obtained by contracting some lonely cut-edge in $D$, we have that $x$, the tail of the bag of $u$, is a cut-vertex in $D$ whose removal disconnects all the bags of the internal vertices of $P_i$ from $r$. However, by Lemma~\ref{lem:no-edge-to-head} and the definition of a bipath we have that $x$ has an inneighbor in the bag of $v_2$. Then Rule~\ref{r:cutneighbor} would apply to $x$, a contradiction.
\end{proof}

Weak bipaths $P_1,\ldots,P_q$ given by Lemma~\ref{lem:decomposition} are called {\em{maximal bipaths}}. Note that for every such maximal bipath $P=v_1,v_2,\ldots,v_p$ and every $j=2,\ldots,p-1$, bag $B_{v_j}$ is linked to $B_{v_{j-1}}$ and $B_{v_{j+1}}$, and to no other bag.

\subsection{New lower bounds on the number of leaves}

In this section our goal is to establish a number of lower bounds on the number of leaves.
Each of the lower bounds is a linear function of a number of some type of vertices or structures in $D$.
These bounds will help us prove that Rule~\ref{r:accept} is correct.
Indeed, to this end it suffices to focus on a no-instance and prove that it has at most $ck$ vertices.
Hence, if we know that $\ml(D)$ is large when there are many vertices of some kind A, then we know that in our no-instance there are few vertices of kind A.
In other words vertices of type A are ``easy''.
In the next section we will show that because of sparsity arguments the number of the remaining vertices (not corresponding to an ``easy type'') is linear in the number of ``easy'' vertices.

In fact, instead of looking for ``easy'' vertices in $D$, we focus of $\contractd$. 
 This is justified by the fact that by Lemma~\ref{lem:bag-size-2} we have $|V(D)| \le 2|V(\contractd)|$, so if we prove that $|V(\contractd)|=O(k)$ then also $|V(D)|=O(k)$. 
 Moreover, the following lemma shows that a lower bound on $\ml(\contractd)$ imply the same lower bound on $\ml(D)$.

\begin{lemma}
\label{lem:contract}
Let $D$ be a connected digraph, and let $D'$ be the digraph obtained from $D$ by contracting a cut-edge. 
Then $\ml(D)\ge\ml(D')$.
\end{lemma}

\begin{proof}
Let $(u,v)$ be the contracted cut-edge and let $x$ be the resulting vertex in $D'$.
Consider any outbranching $T'$ of $D'$.
Then let $T$ be obtained by the following procedure: 1) remove $x$ from $T'$, 2) add vertices $u$ and $v$, 3) add edge $(u,v)$, 4) if $p$ is the parent of $x$ in $T'$, add edge $(p,u)$, 5) for every edge $(x,y)\in E(T')$, add $(u,y)$ to $T$ if $(u,y)\in E(D)$ and add $(v,y)$ otherwise.
Then clearly $T$ is an outbranching with at least the same number of leaves as $T'$. 
Hence it suffices to show that $T$ is a subgraph of $D$.
Otherwise, $(p,u)\not\in E(D)$. 
However, then $(p,v)\in E(D)$.
Also, in $T'$ there is a path from the root to $p$ that avoids $x$. It follows that this path, extended by the edge $(p,v)$ is contained also in $D$, $(u,v)$ is not a cut-edge, a contradiction.
\end{proof}

Since all heads of cut-edges have indegree $1$, and contraction of lonely cut-edges cannot spoil this property for other cut-edges, we infer that every cut-edge of $D$ remains a cut-edge in the process of obtaining $\contractd$ from $D$ by contracting lonely cut-edges one by one. This yields the following.

\begin{corollary}
\label{cor:leaves-D'}
 $\ml(D)\ge \ml(\contractd)$.
\end{corollary}

\heading{A bound on special vertices} Daligault and Thomass\'{e}~\cite{daligault09} show the following lower bound.

\begin{theorem}[\cite{daligault09}]\label{th:2conSpec}
Let $D$ be a $2$-connected rooted digraph. Then $\ml(D)\geq \frac{|\sp(D)|}{30}$.
\end{theorem}

Unfortunately, $\contractd$ is not necessarily $2$-connected so we cannot use the above bound.
However, we can generalize Theorem~\ref{th:2conSpec} as follows.

\begin{theorem}\label{th:2edgeconSpec}
Let $D$ be a connected rooted digraph such that every cut-edge is branching. Then $\ml(D)\geq \frac{|\sp(D)|}{30}-\cv(D)$ and $\ml(D) \geq \frac{|\sp(D)|}{60}$.
\end{theorem}

To prove Theorem~\ref{th:2edgeconSpec}, we first need two lemmas.
By {\em duplicating} a vertex $v$ in a digraph $D$ we mean creating a new digraph $D'$ with $V(D')=V(D)\cup\{v'\}$ and $E(D')=E(D)\cup\{(x,v')\ :\ (x,v)\in E(D)\}\cup\{(v',x)\ :\ (v,x)\in E(D)\}$.

\begin{lemma}\label{prop:double}
Let $D$ be a digraph, and let $D'$ be the digraph obtained from $D$ by duplicating a vertex $v$. 
Then $|\sp(D')|\geq |\sp(D)|$, and $\ml(D)\geq \ml(D')-1$.
\end{lemma}

\begin{proof}
Every special vertex in $D$ is still a special vertex in $D'$: duplicating can neither decrease the in-degree of a vertex, nor remove a simple in-edge. 
Hence $|\sp(D')|\geq |\sp(D)|$.
Take a rooted maximum leaf outbranching $T'$ in $D'$. By symmetry, suppose that $v$ is not a descendant of $v'$.
Let $T = T' \setminus \{v'\} \cup \{(v,w)\ :\ (v',w)\in E(T')\}$.
Note that $T$ is an outbranching in $D$.
If both $v$ and $v'$ were leaves in $T'$ then $v$ is a leaf in $T$, so $T$ has one leaf less than $T'$.
Otherwise even if $v$ is not a leaf in $T$ the number of leaves drops by at most one. 
This finishes the proof.
\end{proof}

\begin{lemma}\label{cl:2edgecon-c}
In any digraph $D$ such that rules 1--4 do not apply and every cut-edge is branching we have $\ml(D)\geq \cv(D)+1$.
\end{lemma}

\begin{proof}
Let $T$ be the spanning tree of $D$ obtained through a Breadth-First-Search started in $r$. 
Consider any cut-vertex $u$. 
Since $u$ is a cut-vertex, we have $|P(u)|\geq 1$. 
If $|P(u)|=1$, let $v$ be the only private neighbor of $u$. 
The edge $(u,v)$ is then a lonely cut-edge, a contradiction. 
Therefore $|P(u)|\geq 2$. 
By Lemma~\ref{lem:indeg} all the edges from $u$ to $P(u)$ are cut-edges.
It follows that every cut-vertex in $D$ has at least two out-neighbors in $T$. 
Hence $T$ has at least $\cv(D)+1$ leaves.
\end{proof}

We are now ready to prove Theorem~\ref{th:2edgeconSpec}.

\begin{proof}[Proof of Theorem~\ref{th:2edgeconSpec}]
Let $S$ be the set of cut-vertices in $D$. 
Take the digraph $D'$ obtained from $D$ by duplicating each vertex in $S$ (in any order). 
We claim that $D'$ is 2-connected. 
Indeed, assume that $D'$ contains a cut-vertex $u$; since a vertex and its duplicate are twins, we can assume that $u\in D$.
Since $D$ is a subgraph of $D'$, it follows that $u$ is a cut-vertex in $D$. 
Now $D'$ contains a duplicate $u'$ of $u$, so every vertex reachable from $r$ in $D'$ is still reachable in $D'-u$, a contradiction. 
Therefore $D'$ is 2-connected and by Theorem~\ref{th:2conSpec} we get $\ml(D')\geq \frac{|\sp(D')|}{30}$. 
By Lemma~\ref{prop:double}, we have $\ml(D) \geq \ml(D')-\cv(D)$ and $|\sp(D')| \geq |\sp(D)|$. 
Thus
\begin{equation}\label{eq1}
\ml(D) \geq \frac{|\sp(D)|}{30}- \cv(D).
\end{equation}
From Lemma~\ref{cl:2edgecon-c} we have
\begin{equation}\label{eq2}
\ml(D) \geq \cv(D)+1.
\end{equation}
The claim now follows from adding \eqref{eq1} and \eqref{eq2}.
\end{proof}

Now it suffices to show that Theorem~\ref{th:2edgeconSpec} can be applied to graph $\contractd$.

\begin{lemma}\label{lem:twopaths}
Suppose $D$ is a rooted digraph that is connected. Then for any vertex $u\neq r$ that is not the head of a cut-edge, one can find two simple paths $P_1,P_2$ from $r$ to $u$ that end with different edges.
\end{lemma}
\begin{proof}
Let $R$ be the set of inneighbors of $u$ that are reachable from $r$ in $D-u$. Since $D$ is connected we have $R\neq \emptyset$, and if $|R|\geq 2$ then we would be done. Suppose therefore that $R=\{v\}$ for some vertex $v$ such that $(v,u)\in E(D)$. Then $(v,u)$ would be a cut-edge, a contradiction.
\end{proof}

Lemma~\ref{lem:twopaths} will be most often used in the following setting. Suppose that we know that in $D$ the head of every cut-edge has indegree $1$. Then if we know that some edge $(v,u)$ is not a cut-edge, then $u$ is not the head of any cut-edge, and hence we can apply Lemma~\ref{lem:twopaths} to it.

\begin{lemma}
 \label{lem:contract-no-new-cut-edges}
 Assume that rules 1--4 do not apply to $D$.
 Let $S$ be the set of lonely cut-edges in $D$.
 Consider any subset $S' \subseteq S$. 
 Let $D_1$ be the graph obtained from $D$ by contracting all edges of $S'$.
 Then $D_1$ does not contain a new cut-edge.
\end{lemma}

\begin{proof}
  Induction on $|S'|$. The claim is trivially true for $|S'|=0$. Assume $|S'|>0$.
  Pick any cut-edge $(x,y)\in S'$ and let $D_0$ be the graph obtained from $D$ by contracting all edges of $S'\setminus\{(x,y)\}$; obviously $D_0$ is connected.
  From the induction hypothesis we have that the set of cut-edges of $D_0$ is a subset of the set of cut-edges of $D$, and hence from the fact that in $D$ all the heads of cut-edges have indegrees equal to $1$, the same conclusion follows for $D_0$ as well. Hence, whenever in $D_0$ we conclude that an edge $(u,v)$ is not a cut-edge, then all the edges incoming to $v$ are also not cut-edges.
  We will show that contracting $(x,y)$ in $D_0$ does not create a new cut-edge in $D_1$.

 Assume for a contradiction that $(u,v)$ is a new cut-edge in $D_1$, i.e., either $(u,v)\not\in E(D_0)$ or $(u,v)\in E(D_0)$ and is not a cut-edge in $D_0$.
 In the former case we have two subcases: contracting $(x,y)$ creates vertex $u$ or $v$.
 
 \mycase{1} $v$ is obtained by contracting $(x,y)$.
 Then there is an edge $(u,x)$ or $(u,y)$ in $D_0$. However, the latter situation is impossible because then an edge enters $y$ in $D$, a contradiction with Lemma~\ref{lem:indeg}.
 Hence $(u,x)\in D_0$, and in particular $x\ne r$.
 Edges entering $x$ in $D$ are not cut-edges by Lemma~\ref{lem:head-tail}, and hence by induction hypothesis no cut-edge enters $x$ in $D_0$.
 By Lemma~\ref{lem:twopaths} it follows that in $D_0$ there are two paths $P_1$, $P_2$ from $r$ to $x$, each entering $x$ via a different edge, say $P_1$ by $(a_1,x)$ and $P_2$ by $(a_2,x)$, with $a_1\ne a_2$. Note that $a_1\ne y$ and $a_2\ne y$ because by Rule~\ref{r:cut-double-edge} we have that $(y,x)\notin E(D)$.
 By replacing $(a_1,x)$ with $(a_1,v)$ in $P_1$ and $(a_2,x)$ with $(a_2,v)$ in $P_2$ we get two paths $P_1'$ and $P_2'$ from $r$ to $v$ in $D_1$ that end with different edges.
 It follows that $(u,v)$ is not a cut-edge in $D_1$, a contradiction.
 
 \mycase{2} $u$ is obtained by contracting $(x,y)$.
 Then $D_0$ contains $(x,v)$ or $(y,v)$.
 No other edge leaving $x$ is a cut-edge in $D$ because $(x,y)$ is lonely in $D$.
 Also no edge leaving $y$ is a cut-edge in $D$ by Lemma~\ref{lem:head-tail}.
 Hence by induction hypothesis neither $(x,v)$ nor $(y,v)$ can be a cut-edge in $D_0$.
 Since $v$ has an incoming edge that is not a cut-edge in $D_0$, as explained before we infer that no edge incoming to $v$ in $D_0$ is a cut-edge. 

 From Lemma~\ref{lem:twopaths} it follows that in $D_0$ there are two paths $P_1$, $P_2$ from $r$ to $v$, each entering $v$ via a different edge.
 If $(u,v)$ is a cut-edge in $D_1$, then it means that $P_1$ ends with $(x,v)$ and $P_2$ ends with edges $(x,y),(y,v)$, because $(x,y)$ is a cut-edge.
 If $v\in D$ then Rule~\ref{r:cutneighbor} would apply to $D$ (with $v$ as $x$), a contradiction.
 Otherwise $v$ is obtained by contracting a cut-edge $(v_1,v_2)$.
 However, by Lemma~\ref{lem:no-edge-to-head}, no other edge enters $v_2$ in $D$, so $D$ contains both edges $(x,v_1)$ and $(y,v_1)$.
 Again, we see that Rule~\ref{r:cutneighbor} applies to $D$ (with $v_1$ as $x$), a contradiction.
 
 \mycase{3} Neither $u$ nor $v$ is obtained by contracting $(x,y)$. Since $(u,v)$ is not a cut-edge in $D_0$, as in the previous case we infer that in fact no edge incoming to $v$ is a cut-edge in $D_0$.
 By Lemma~\ref{lem:twopaths}, in $D_0$ there are two paths $P_1$ and $P_2$ from $r$ to $v$, each ending with a different edge.
 Let us assume that $P_1$ ends by $(a,v)$ and $P_2$ ends by $(b,v)$, for some $a\ne b$.
 Let $P_1'$ and $P_2'$ be the paths in $D_1$ obtained from $P_1$ and $P_2$ by contracting edge $(x,y)$, and possibly omitting a loop in case both $x$ and $y$ were traversed by $P_1$ or $P_2$.
 Then $P_1'$ and $P_2'$ end with different edges unless $\{x,y\}=\{a,b\}$. By symmetry suppose that $(x,y)=(a,b)$. However, $(x,y)$ is a cut-edge. Hence if $v\in D$ then Rule~\ref{r:cutneighbor} would apply to $D$ (with $v$ as $x$), a contradiction. Otherwise $v$ is obtained by contracting a cut-edge $(v_1,v_2)$, and the same reasoning as in the previous case also gives a contradiction.
\end{proof}

\begin{lemma}
\label{lem:D'-2-edge-conn}
 If rules 1--4 do not apply to $D$, then in graph $\contractd$ all cut-edges are branching.
\end{lemma}

\begin{proof}
 By Lemma~\ref{lem:contract-no-new-cut-edges} applied to all lonely cut-edges, in $\contractd$ all lonely cut-edges are contracted and no new cut-edges appear.
 Moreover all the cut-edges that are branching in $D$ are also cut-edges in $\contractd$ (since indegrees of their heads are 1), so they are also branching in $\contractd$.
 This finishes the proof.
\end{proof}

By Lemma~\ref{lem:D'-2-edge-conn} and Corollary~\ref{cor:leaves-D'} we get the following lower bound.

\begin{lemma}
\label{lem:lower-bound-special}
 $\displaystyle \ml(D)\ge \frac{|\sp(\contractd)|}{60}.$
\end{lemma}

\heading{A bound on isolated vertices}
We say that a bag $B$ is {\em special} when $v_B$ is special in $\contractd$.
We say that a bag $B$ is {\em isolated} when $B$ is a non-special bag of size 2 and there is no edge from $t_B$ to a special bag.
Vertex $v\in V(\contractd)$ is {\em isolated} if $v=v_B$ for some isolated bag $B$. 

The set of all isolated vertices in $\contractd$ is denoted by $\iso(\contractd)$.

By {\em shortcutting} a vertex $v\ne r$ in a digraph $D$ we mean creating a new digraph $D'$ obtained from $D$ by removing $v$ and adding an edge $(x,y)$ for every directed path $(x,v,y)$ in $D$.

Let $\shortcutd$ be the graph obtained from $D$ by (i) contracting all lonely cut-edges that form a non-isolated bag, and then (ii) shortcutting every tail of an isolated bag. Note that $\shortcutd$ is not necessarily $H$-minor-free, but we will use it only as an auxiliary construction when establishing a lower bound on $\ml(D)$ in terms of $\iso(\contractd)$.

The proof of following lemma can be found in~\cite{daligault09}:

\begin{lemma}\label{lem:shortcut}
Let $D$ be a digraph, and let $D'$ be the digraph obtained from $D$ by shortcutting a cut-vertex $v$. 
Then $\ml(D)=\ml(D')$.
\end{lemma}

Lemmas~\ref{lem:shortcut} and~\ref{lem:contract} imply the following

\begin{equation}
\label{eq:leaves-tilde-g}
 \ml(D)\ge\ml(\shortcutd).
\end{equation}

We also observe the following property.

\begin{lemma}
 \label{lem:shortcut-no-new-cut-edges}
 Suppose $D$ is a connected rooted digraph where every head of a cut-edge has indegree $1$. Let $u$ be a vertex and suppose that $r\notin N^-(u)$ and there is no vertex $v\in N^-(u)$ such that $u$ becomes disconnected from $r$ after removing $v$.
 Then after shortcutting $u$ no new cut-edges appear in $D$.
\end{lemma}

\begin{proof}
 Note that the assumption of the lemma implies that in $D$ there is no cut-edge that enters $u$, so we can apply Lemma~\ref{lem:twopaths} to $u$.
 Let $D$ and $D'$ denote the graph before and after shortcutting $u$.
 Assume that a new cut-edge $(x,y)$ appears in $D'$.
 
 \mycase{1} $(x,y)\in E(D)$ and $(x,y)$ is not a cut-edge in $D$. Since every head of a cut-edge has indegree $1$, we infer that no cut-edge enters $y$.
 By Lemma~\ref{lem:twopaths}, in $D$ there are two paths $P_1$ and $P_2$ from $r$ to $y$, ending with different edges $e_1$ and $e_2$.
 Let $P_1'$ and $P_2'$ be the paths obtained from $P_1$ and $P_2$ by shortcutting $u$.
 If $P_1'$ and $P_2'$ end with the same edges as $P_1$ and $P_2$, then $(x,y)$ is not a cut-edge in $D'$, a contradiction.
 Otherwise observe that exactly one of $P_1'$ and $P_2'$ has changed the last edge, because otherwise $e_1=e_2=(u,y)$.
 By symmetry assume $e_1=(u,y)$ and $e_2=(w,y)$, for some $w\ne u$. 
 Then $P_1'$ ends with $(w,y)$ and $(w,u)$ is the second last edge of $P_1$, or otherwise we are done.
 By the assumption of the lemma, removal of $w$ does not disconnect $u$ from $r$, so there is a path $Q$ from $r$ to $u$ that avoids $w$. If this path traverses $y$, then its prefix is a path from $r$ to $y$ in $D'$ that enters $y$ from a different vertex than $w$. Otherwise after prolonging $Q$ with $(u,y)$ and shortcutting $u$ we obtain a path from $r$ to $y$ in $D'$ that enters $u$ from a different vertex than $w$. In both cases we obtained two paths from $r$ to $y$ in $D'$ that end with different edges, which means that no edge incoming to $y$ can be a cut-edge. This is a contradiction with $(x,y)$ being a cut-edge.

 \mycase{2} $(x,y)\not\in E(D)$, i.e., $(x,y)$ is obtained by shortcutting $u$ and $(x,y)$ was not present in $D$.
 By Lemma~\ref{lem:twopaths}, in $D$ there are two simple paths $P_1$ and $P_2$ from $r$ to $u$, ending with different edges $(a,u)$ and $(b,u)$, for some $a\ne b$. If any of these paths traverses $y$, then some its prefix is a path in $D'$ from $r$ to $y$ that avoids the new edge $(x,y)$, due to $(x,y)$ being not present in $D$. This is a contradiction with $(x,y)$ being a cut-edge. Suppose then that neither $P_1$ nor $P_2$ traverses $y$; in particular $a\ne y$ and $b\ne y$. Then by replacing $(a,u)$ by $(a,y)$ and $(b,u)$ by $(b,y)$ we get two paths in $D'$ from $r$ to $y$ ending by different edges, so $(x,y)$ is not a cut-edge, a contradiction. 
\end{proof}

Let $S$ be the set comprising $r$ and all the special vertices of $\contractd$. Let us invoke Lemma~\ref{lem:decomposition} on the set $S$, and thus obtain a family of maximal weak bipaths $P_1,P_2,\ldots,P_q$ with properties as in this lemma.

Consider the process of creating $\shortcutd$. 
After contracting all lonely cut-edges corresponding to non-isolated bags, by Lemma~\ref{lem:contract-no-new-cut-edges}, no new cut-edges appear. We would like to derive the same conclusion for $\shortcutd$ as well, however we must be careful due to the non-trivial prerequisites of Lemma~\ref{lem:shortcut-no-new-cut-edges}.

\begin{lemma}
\label{lem:2-edge-conn}
If reduction rules do not apply to $D$ then every cut-edge in graph $\shortcutd$ is branching.
\end{lemma}
\begin{proof}
Let $D'$ be the graph after contracting the lonely cut-edges corresponding to non-isolated bags. As argued above, from Lemma~\ref{lem:contract-no-new-cut-edges} if follows that $D'$ has no new cut-edge, i.e., all cut-edges of $D'$ are either original branching cut-edges of $D$, or original lonely cut-edges of $D$ that correspond to isolated bags.

In $\contractd$, every isolated vertex is some internal vertex on one of the bipaths $P_i$. We can view the construction of $\shortcutd$ from $D'$ as follows: We iterate through the bipaths $P_1,P_2,\ldots,P_q$ one by one. For each of them, we iterate through the internal vertices $w$ of the bipath from left to right, and in $D'$ we shortcut the tail of the bag corresponding to $w$ provided this bag is isolated. We prove now that during this process we maintain the following invariant: 
\begin{itemize}
\item[(1)] No new cut-edge has been created, and in particular all the heads of cut-edges in the current digraph have indegree $1$.
\item[(2)] For every $v\in \contractd$ that is an isolated vertex on some weak bipath, and $t_v$ is the tail of its bag, the following holds: as long as $t_v$ is not yet shortcutted, in $D'$ there is no inneighbor of $t_v$ which is a cut-vertex whose removal disconnects $t_v$ from the root.
\end{itemize}

We now show that invariant (2) holds for every such $v$ throughout the process, up to the point when $t_v$ is shortcutted. Let us fix $v$, and suppose $v=v_i$ lies on a maximal weak bipath $P_\alpha=v_1,v_2,\ldots,v_p$ in $\contractd$, for some $\alpha\in \{1,2,\ldots,q\}$. For simplicity, denote $B_j=B_{v_j}$. By Lemmas~\ref{lem:root} and~\ref{lem:decomposition}, $v_1,v_p\in S$, $r\notin\{v_1,v_p\}$, and $v_1\neq v_p$. Note that vertices $v_1,v_p$ are already present in $D'$. Let $W$ be the set of (a) all vertices of $D'$ that are contained in bags $B_i$, for $i=2,\ldots,p-1$, and (b) all vertices $v_i$, for $i=1,2,\ldots,p$, for which $v_i\in D'$.


We now claim that in $D'$ there is a path $Q_1$ from $r$ to $v_1$ that avoids the vertices of $W$. Indeed, if $v_1$ was disconnected from $r$ in $D'-W$, then any path from $r$ to $v_1$ would need to use the unique edge from $v_p$ to $B_{p-1}$ (or $v_{p-1}$), so this edge would be a cut-edge in $D'$. This is a contradiction, because this edge was not a cut-edge in $D$, since cut-edges of $D$ not residing in one bag must be branching and the head of each branching cut-edge in $D$ is special in $\contractd$. Similarly, there is a path $Q_2$ from $r$ to $v_p$ that avoids $W$.

From Lemmas~\ref{lem:no-edge-to-head} and~\ref{lem:one-edge} it follows that in $D'$ there is a path $R_1$ from $v_1$ to $t_v$ that traverses consecutive bags $B_1,B_2,\ldots,B_{i-1},B_i$ (possibly contracted when constructing $D'$), and in each it visits either only the tail, or first the tail and then the head. Similarly, there is a path $R_2$ from $v_p$ to $t_v$ that traverses consecutive bags $B_p,B_{p-1},\ldots,B_{i+1},B_i$ (possibly contracted when constructing $D'$), and in each it visits either only the tail, or first the tail and then the head. In particular, since $v_1\neq v_p$, we have that $R_1$ and $R_2$ are vertex-disjoint apart from the last vertex $t_v$. 

Let $M_1$ be the concatenation of $Q_1$ and $R_1$, and similarly define $M_2$. We now examine what happens with paths $M_1$ and $M_2$ during the process of obtaining $\shortcutd$ from $D'$. Every shortcutting of a vertex gives rise to a natural transformation of simple paths in $D'$, where the traversal of the shortcutted vertex is replaced by the usage of a newly introduced edge. Observe that the prefix $Q_1$ can only get shortcutted during the process, and similarly holds for the prefix $Q_2$. However, $v_1$ and $v_p$ are not being shortcutted. Finally, the internal vertices of both $R_1$ and $R_2$ also can get shortcutted, but we maintain the invariant that these suffixes remain vertex-disjoint.

Concluding, during the process of obtaining $\shortcutd$ from $D'$, $M_1$ and $M_2$ are always two paths from $r$ to $t_v$, and their suffixes beginning from $v_1$ and $v_p$ are always vertex-disjoint apart from the last vertex. Moreover, the vertices appearing before $v_1$ on $M_1$ cannot become the inneighbors of $t_v$ during the shortcutting process due to not belonging to $W\cup \{v_1,v_p\}$, and the symmetrical claim holds for $M_2$ as well. We conclude that at any moment of the process, the removal of any inneighbor of $t_v$ cannot affect both paths $M_1$ and $M_2$ at the same time. 

Hence invariant (2) holds throughout the process. Invariants (1) and (2) are exactly the prerequisites of Lemma~\ref{lem:shortcut-no-new-cut-edges} when applied to shortcutting $t_v$. Hence, by iteratively applying Lemma~\ref{lem:shortcut-no-new-cut-edges} we conclude that no new cut-edge appears in $\shortcutd$, and in particular every application shows that invariant (1) is maintained in the next step. Therefore, the cut-edges of $\shortcutd$ are simply the branching cut-edges of the original digraph $D$.
\end{proof}


Motivated by Lemma~\ref{lem:2-edge-conn} and Equation~\eqref{eq:leaves-tilde-g} we are going to show that if there are many isolated vertices in $\contractd$, then 
there are many special vertices in $\shortcutd$, which, together with Theorem~\ref{th:2edgeconSpec}, implies the desired lower bound. Note that every non-special (in particular, every isolated) vertex in $\contractd$ is an internal vertex of some weak bipath $P_i$, and hence a non-special bag is linked to exactly two other bags --- neighbors on the bipath.

\begin{lemma}\label{lem:removable-bags}
Assume reduction rules do not apply to $D$.
Suppose bag $A$ is isolated.
Then $h_A$ is special in $\shortcutd$ or there is a non-special bag $B$ linked to $A$ such that $h_B$, or $v_B$ if $B$ gets contracted, is special in $\shortcutd$ .
\end{lemma}

\begin{proof}
 Since $A$ is isolated, there is some bipath $P_i=v_1,v_2,\ldots,v_p$ such that $v_A=v_a$ for some $2\leq a\leq p-1$. Denote $B_i=B_{v_i}$. Since Rule~\ref{r:cut-edge} does not apply, we infer that $d^+(t_A) \ge 2$. One of these edges goes to $h_A$, whereas the second needs to go to one of the two neighboring bags on $P$, because $A$ is isolated. By symmetry, suppose that there is an edge from $t_A$ to $B=B_{a+1}$. Of course, $B$ is linked to $A$ and $B$ is not special, because there is an edge from $t_A$ to $B$ and $A$ is isolated. We consider two cases regarding the size of $B$.
 
 \mycase{1} $|B|=2$. We will show that at least one of $h_A$, $h_B$ is special in $\shortcutd$ .
 By Lemma~\ref{lem:no-edge-to-head}, the edge from $A$ to $B$ is $(t_A,t_B)$.
 Then by Rule~\ref{r:two-cut-edges}, $(t_B,t_A) \notin E(D)$. 
 By Lemma~\ref{lem:no-edge-to-head} it follows that $(h_B,t_A)\in E(D)$.
 By Lemma~\ref{lem:one-edge}, $(t_A,t_B)$ and $(h_B,t_A)$ are the only edges between $A$ and $B$.
 
 Let $t$ be the minimum index $i<a$ such that $B_{i+1}, \ldots, B_a$ are all isolated and there is an edge from $t_{B_j}$ to $t_{B_{j+1}}$ for each $i<j<a$.
 By the minimality of $t$ and Lemma~\ref{lem:one-edge}, it follows that either $B_t$ is not isolated or there is an edge from $h_{B_t}$ to $t_{B_{t+1}}$.
 In either case, $\shortcutd$ has an edge $e$ incoming to $h_B$ (or $v_B$, if $B$ gets contracted) from a vertex corresponding to bag $B_t$ (i.e., either from $v_{B_t}$ or $h_{B_t}$).
 If in $\shortcutd$ there is no edge from $h_B$ (or $v_B$) to a vertex that corresponds to $B_t$, then $h_B$ ($v_B$) is special in $\shortcutd$, and we are done.
 So assume that there is such an edge.
 It means that in $D$ there must be an edge from $t_A$ to $t_{B_{a-1}}$. Then $t=a-1$, because otherwise Rule~\ref{r:two-cut-edges} would apply.
 By Lemma~\ref{lem:one-edge} in $D$ there is no edge from $h_A$ to $B_{a-1}$. We argued earlier that there is also no edge in $D$ from $h_A$ to $B_{a+1}=B$.
 It follows that $(h_A,h_B)\notin E(\shortcutd)$. However, after shortcutting $A$ we get $(h_B,h_A)\in E(\shortcutd)$.
 Hence, $h_A$ is special in $\shortcutd$.
 
 \mycase{2} $|B|=1$.
 By Lemmas~\ref{lem:no-edge-to-head} and \ref{lem:one-edge}, the only edges between $A$ and $B$ are $(t_A,t_B)$ and $(t_B,t_A)$.
 Then $\shortcutd$ contains edge $(t_B,h_A)$.
 If $(h_A,t_B) \notin\shortcutd$, then $h_A$ is special in $\shortcutd$ and we are done.
 Otherwise, denoting $C=B_{a-1}$, it must hold that $C$ is isolated (so in particular non-special) and there must be edges $(h_A,t_{C}),(t_{C},t_A)$ in $D$.
 Then by the same argument as in Case 1 (with $C$ playing the role of $A$ and $A$ playing the role of $B$), $h_A$ or $h_C$ is special in $\shortcutd$.
 This ends the proof.
\end{proof}

\begin{lemma}
 \label{lem:lower-bound-isolated}
 If reduction rules do not apply to $D$ then $\ml(D)\ge\tfrac{|\iso(\contractd)|}{180}$.
\end{lemma}

\begin{proof}
 By Lemmas~\ref{lem:contract} and~\ref{lem:shortcut}, $\ml(D)\ge\ml(\shortcutd)$.
 By Lemma~\ref{lem:2-edge-conn} and Theorem~\ref{th:2edgeconSpec} we get $\ml(\shortcutd)\ge\frac{|\sp(\shortcutd)|}{60}$.
 By Lemma~\ref{lem:removable-bags}, to every isolated bag $A$ we can assign a non-special bag $B$, such that $h_B$ is special in $\shortcutd$ and either $B=A$ or $B$ is linked to $A$. 
 By the definition, there are at most two bags linked to a non-special bag (corresponding to the neighbors of $v_B$ on a weak bipath in $\contractd$). 
 It follows that $|\sp(\shortcutd)|\ge\frac{|\iso(\contractd)|}3$.
 Together with the previous inequalities this implies $\ml(D)\ge\tfrac{|\iso(\contractd)|}{180}$.
\end{proof}

We will say that a vertex $v$ of $\contractd$ is {\em easy} when $v=r$, or $v$ is special, or $v$ is isolated in $\contractd$.
A vertex that is not easy is called {\em hard}. We now invoke once more Lemma~\ref{lem:decomposition}, but this time instead of $S$ we take the set of all the easy vertices. Every maximal bipath obtained in this decomposition will be called a {\em{maximal hard bipath}}. In other words, a weak bipath in $\contractd$ is {\em hard} if all its internal vertices are hard.
The sets of all easy and hard vertices in $\contractd$ are denoted by $\ea(\contractd)$ and $\hd(\contractd)$, respectively.
For any maximal hard bipath $P'$ in $\contractd$ we define $O(P') = N^+_{\contractd}(V(P')\setminus\{u,v\})$, where $u$ and $v$ are the extremities of $P'$.

\heading{A bound on slaves} 
For every pair of easy vertices $u,v\in\ea(\contractd)$ and a subset $S\subseteq V(\contractd)$ with $\{u,v\}\subseteq S$, if there is a hard bipath $P'$ between $u$ and $v$ such that $O(P')=S$,
we choose arbitrarily two such paths (or one, if only one exists) and we call them {\em masters}, while all the remaining hard bipaths $P''$ between $u$ and $v$ with $O(P'')=S$ are called {\em slaves} of respective masters, or just {\em{slaves}}.
The number of all slaves in $\contractd$ is denoted by $\sl(\contractd)$.

\begin{lemma}
 \label{lem:slaves}
$\ml(D) \ge \sl(\contractd)$.
\end{lemma}

\begin{proof}
By Lemma~\ref{lem:contract} it suffices to show that $\ml(\contractd) \ge \sl(\contractd)$.
We will show that in fact every outbranching $T$ of $\contractd$ has at least $\sl(\contractd)$ leaves.

Fix an arbitrary outbranching $T$ of $\contractd$.
It is easy to see that in any outbranching $T$, the number of leaves is equal to $1+\sum_{u\in V(T)} \max(\deg^+_T(u)-1,0)$.
Consider a slave $Z=v_1,\ldots,v_{\ell}$ with $O(Z)=S$ and extremities $v_1,v_\ell\in S$, and let $M_1,M_2$ be its masters. Then either $(v_1,v_2)\in E(T)$, or $(v_{\ell},v_{\ell-1})\in E(T)$. Let slave $Z$ charge vertex $v_1$ in the former case, and charge vertex $v_\ell$ in the latter case. Also on $M_1$ and $M_2$ at least one edge outgoing from $v_1$ and one edge outgoing from $v_\ell$ is present in $T$.  We conclude that the total contribution to the outdegrees in $T$ of $v_1$ and $v_\ell$ from $M_1,M_2$ and their slaves is at least the number of times $v_1$ and $v_\ell$ are charged by the slaves of $M_1,M_2$, plus $2$ for $M_1$ and $M_2$.

Let $X\subseteq \ea(\contractd)$ be the set of easy vertices that are the extremities of some slave. Then $1+\sum_{u\in V(T)} \max(\deg^+_T(u)-1,0)\geq (\sum_{u\in X} \deg^+_T(u))-|X|$. On the other hand, from what we argued in the previous paragraph it follows that $\sum_{u\in X} \deg^+_T(u)\geq \sl(\contractd)+2|F|$, where $F$ is the set of equivalence classes of slaves partitioned according to their masters. However, since every bipath has two extremities, it follows that $|X|\leq 2|F|$. Hence $\sum_{u\in X} \deg^+_T(u)-|X|\geq \sl(\contractd)$ and $T$ has at least $\sl(\contractd)$ leaves.

\end{proof}

\subsection{The size bound}
\label{sec:analysis}

In this section we prove the following theorem which imply the correctness of Rule~\ref{r:accept}.

\begin{theorem}
\label{thm:size-bound}
Let $H$ be a graph. Let $D$ be an $H$-minor-free digraph such that rules 1--\ref{r:cut-double-edge} do not apply. If $\ml(D)< k$, then $|V(D)|= 2^{O(|H|\sqrt{\log|H|})}k$.
\end{theorem}

Throughout the section we assume that rules 1--\ref{r:cut-double-edge} do not apply to $D$.
The results from the previous section give a bound of $O(k)$ on the number of easy vertices.
Our plan in this section is to show a linear bound on the number of hard vertices in terms of $|\ea(\contractd)|+\sl(\contractd)$ and next get a bound on $|V(D)|$ as a corollary.

It follows that our task is to show that the total length of hard weak bipaths in $\contractd$ is not too large.
Let us state a few useful properties of such bipaths.

\begin{lemma}
\label{lem:no-long-proper}
 Let $\ell\ge 9$ and let $P'=v_1,\ldots,v_{\ell}$ be a hard bipath in $\contractd$ such that $v_1$ and $v_{\ell}$ are easy.
 For every $i=3,\ldots,\ell-6$ there is at least one edge in $D$ from $t_{B_{v_j}}$, for some $j=i,\ldots,i+4$, to a vertex outside $\cup_{j'=2}^{\ell-1}B_{v_{j'}}$.
\end{lemma}

\begin{proof}
 Fix $i\in\{3,\ldots,\ell-6\}$ and consider the length 4 bipath $v_i,\ldots,v_{i+4}$.
 For convenience denote $B_j=B_{v_j}$.
 If for some $j=i+1,i+2,i+3$ there is an edge from $B_j$ with head not in $B_{j-1}\cup B_{j+1}$, then by Lemma~\ref{lem:decomposition}$(iii)$ this head is outside $\cup_{j'=2}^{\ell-1}B_{v_{j'}}$ and we are done.
 Hence the edges leaving $B_{i+1}$, $B_{i+2}$, and $B_{i+3}$ go only to the neighboring bags.  
 Since Rule~\ref{r:bipath} does not apply, for some $j=i,\ldots,i+4$ the bag $B_j$ is of size~2.
 Since $v_j$ is hard, $B_j$ is not isolated.
 Hence, there is an edge $e$ in $D$ from $t_{B_j}$ to a special bag $B$. 
 Since $v_2,\ldots,v_{\ell-1}$ are hard, $B$ is none of $B_2,\ldots,B_{\ell-1}$
\end{proof}

\begin{lemma}
\label{lem:many-neighbors}
 For any maximal hard weak bipath $P'$ in $\contractd$, we have $|\hd(\contractd)\cap V(P')|\le 10|O(P')|+6$.
\end{lemma}

\begin{proof}
Let $P'=v_1,\ldots,v_{\ell}$.
We can assume that $\ell\ge 9$, for otherwise $|\hd(\contractd)\cap V(P')|\le 6$ and the claim holds trivially.
For convenience denote $B_i=B_{v_i}$.
By Lemma~\ref{lem:no-long-proper} there are at least $\lfloor \frac{\ell-4}{5} \rfloor$ edges from tails of bags $B_3,\ldots,B_{\ell-2}$ to vertices outside $\cup_{i=2}^{\ell-1}B_{v_j}$.
Let $Z$ denote the set of these edges.
We claim that for every vertex $u\in V(D)$ there are at most two edges from $Z$ with heads in $u$.
Indeed, assume that $u$ has got three in-neighbors $t_{B_a}, t_{B_b}, t_{B_c}$ in $D$, with $a<b<c$. 
Then $N^-(u)\setminus \{t_{B_b}\}$ cuts $t_{B_b}$ (and all vertices of $B_{a+1},\ldots,B_{c-1}$) from $r$, a contradiction to the fact that $D$ is reduced with respect to Rule~\ref{r:cutneighbor}. 
Hence the edges in $Z$ have at least $\lfloor\frac{\ell-4}5 \rfloor \cdot \frac12 \ge \frac{\ell-8}5\cdot \frac12$ different heads. 
By Lemma~\ref{lem:no-edge-to-head} these heads are tails of bags, and by Lemma~\ref{lem:bag-size-2} each of them corresponds to a different vertex in $\contractd$.
It follows that the vertices $v_3,\ldots,v_{\ell-2}$ have in $\contractd$ at least $\frac{\ell-8}{10}$ neighbors in $O(P')$, so $|O(P')|\geq \frac{\ell-8}{10}$. 
Since $|\hd(\contractd)\cap V(P)|=\ell-2$ it follows that $|\hd(\contractd)\cap V(P)|\le 10|O(P')|+6$.
\end{proof}

In what follows we are going to bound the size of $\contractd$ using its sparsity properties.
To this end we use an auxiliary bipartite graph $G$, called the {\em bipath minor} of $\contractd$, constructed as follows.
We put $V(G)=A \cup B$, where $A=\ea(\contractd)$, and $B$ is the set of all maximal hard bipaths in $\contractd$.
For every maximal hard bipath $P'$ in $\contractd$  with extremities $u,v \in \ea(\contractd)$, the neighborhood of the corresponding vertex in $B$ is exactly $O(P')$.

\begin{lemma}\label{cl:fD-H}
Let $H$ be a graph. If $D$ is $H$-minor-free, then $|\hd(\contractd)| = 2^{O(|H|\sqrt{\log|H|})} (|\ea(\contractd)|+\sl(\contractd))$.
\end{lemma}

\begin{proof}
Consider an arbitrary hard vertex $v$ of $\contractd$.
Consider the maximal hard weak bipath $P'$ in $\contractd$ that contains $v$.
Then $P'$ corresponds to a vertex in $B$ and by Lemma~\ref{lem:many-neighbors}, it has at most $10|O(P')|+6$ internal vertices.
It follows that

\begin{equation}
\label{eq:hd}
|\hd(\contractd)|\leq \sum_{v \in B}(10\deg_G(v)+6)\leq \sum_{v \in B}16\deg_G(v). 
\end{equation}

Note that $G$ is a minor of (the undirected version of) $D$ since it can be obtained from $\contractd$ by edge contractions and deletions, and $\contractd$ in turn is obtained from $D$ by contractions.
Hence, $G$ is $H$-minor-free. 
Moreover, $G$ is simple.
By Lemma~\ref{lem:degeneracy-h-minor}, we know that $G$ is $d_H$-degenerate, for $d_H={O(|H|\sqrt{\log|H|})}$. 
Let $B_m$ and $B_s$ denote the vertices in $B$ for which the corresponding maximal hard bipath is master and slave, respectively.
By~\eqref{eq:hd} we get
\begin{eqnarray}
|\hd(\contractd)| & \leq & 16\sum_{v \in B}\deg_G(v)\nonumber\\
& \leq & 16\hspace{-5mm}\sum_{\substack{v \in B \\ \deg_G(v) > 2d_H}}\hspace{-5mm}\deg_G(v)+16\hspace{-6mm}\sum_{\substack{v \in B_s \\ \deg_G(v) \le 2d_H}}\hspace{-5mm}\deg_G(v)+16\hspace{-5mm}\sum_{\substack{v \in B_m \\ \deg_G(v) \le 2d_H}}\hspace{-5mm}\deg_G(v).\nonumber
\end{eqnarray}

Let us bound each of the terms separately. By Lemma~\ref{lem:big-degrees}, we have 
$$\sum_{\substack{v \in B_s \\ \deg_G(v) > 2d_H}} d(v) \leq 2d_H|A| = O(|H|\sqrt{\log|H|}\cdot |\ea(\contractd)|).$$ 
Obviously, 
$$\sum_{\substack{v \in B_s \\ \deg_G(v) \le 2d_H}}\deg_G(v) \le 2d_H\sl(\contractd) = O(|H|\sqrt{\log|H|} \sl(\contractd)).$$
Finally, \[\sum_{\substack{v \in B_m \\ \deg_G(v) \le 2d_H}}\deg_G(v) = \sum_{\substack{S \subseteq A \\ |S|\leq 2 d_H}}\hspace{-2mm}|S| \cdot |\{ v \in B_m \colon N_G(v)=S\}|\le 2 d_H \hspace{-2mm}\sum_{\substack{S \subseteq A \\ |S|\leq 2 d_H}} |\{ v \in B_m \colon N_G(v)=S\}|.\]
By Corollary~\ref{cor:boundedsmallneighborhood}, there is a constant $c_H=2^{O(|H|\sqrt{\log|H|})}$ such that there are at most $c_H|A|$ distinct neighborhoods of vertices in $B$. 
For each such neighborhood $S \subseteq A$ and for every pair of vertices $u,v\in S$ there are at most two master bipaths $P'$ with endpoints $u$ and $v$ and such that $O(P')=S$.
Therefore for a fixed neighborhood $S$ of size at most $2d_H$ we have $|\{v \in B_m | N_G(v)=S\}|\leq 2{|S| \choose 2} \leq 2{2 d_H \choose 2}=O(d_H^2)$. 
Hence \[\sum_{{S \subseteq A,  |S|\leq 2 d_H}} |\{ v \in B_m | N_G(v)=S\}|=O(c_H\cdot d_H^2 \cdot |\ea(\contractd)|) = 2^{O(|H|\sqrt{\log|H|})} |\ea(\contractd)|.\]
The claim follows.
\end{proof}

Now we can finish the proof of Theorem~\ref{thm:size-bound}.
Assume $\ml(D)< k$.
By Lemmas~\ref{lem:lower-bound-special} and~\ref{lem:lower-bound-isolated}, $\ea(\contractd) < 60k + 180k$.
Moreover, by Lemma~\ref{lem:slaves}, $\sl(\contractd) < k$.	
This, with Lemma~\ref{cl:fD-H} gives the claim of Theorem~\ref{thm:size-bound}.

\section{\probIOB in graphs of bounded expansion}\label{sec:IOB-kernel}

In this section we give a linear kernel for \probIOBshort on any graph class $\mathcal{G}$ of bounded expansion. To this end, we modify the approach of Gutin, Razgon and Kim~\cite{Gutin-quadratic}. Before we proceed to the argumentation, let us remark that Gutin et al. work with a slightly more general problem, where the root of the outbranching is not prescribed; of course, the outbranching is still required to span the whole vertex set. Note that the variant with a prescribed root $r$ can be reduced to this variant simply by removing all in-arcs of $r$, which forces $r$ to be the root of any outbranching of the given digraph. Since our kernel will be an induced subgraph of $D$ and $r$ will not be removed by any reduction, it will be still true that $r$ is the only candidate for the root of an outbranching. Hence, the resulting instance will be equivalent in both variants. Therefore, from now on we work with variant without prescribed root in order to be able to use the observations of Gutin et al. as black-boxes.

First, Gutin et al. observe that in an instance that cannot be easily resolved, one can find a small vertex cover (of the underlying undirected graph).

\begin{lemma}[\cite{Gutin-quadratic}]
\label{lem:vertex-cover}
Given a digraph $D$, we can either build an out-branching with at least $k$ internal vertices or obtain a vertex cover of size at most $2k-2$ in $O(n^2m)$ time.
\end{lemma}

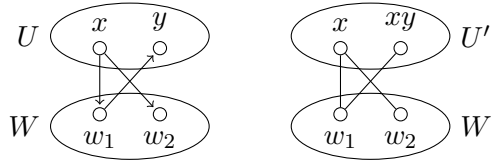
\begin{figure}[h]
 \centering
\begin{tikzpicture}[scale=\scalefactor]
\tikzstyle{whitenode}=[draw,circle,fill=white,minimum size=5pt,inner sep=0pt]
\tikzstyle{set}=[draw,ellipse,fill=white,minimum height=25pt,minimum width=60pt,inner sep=0pt]
 \tikzstyle{blacknode}=[draw,circle,fill=black,minimum size=4pt,inner sep=0pt]
\tikzstyle{texte}=[circle,minimum size=5pt,inner sep=0pt]
\draw (0,0) node[set] (U) [label=left:$U$] {};  
\draw (0,-1.5) node[set] (W) [label=left:$W$] {};

\draw (-0.5,-0.2) node[whitenode] (x) [label=90:$x$] {}; 
\draw (0.5,-0.2) node[whitenode] (y) [label=90:$y$] {}; 

\draw (-0.5,-1.3) node[whitenode] (w1) [label=-90:$w_1$] {}; 
\draw (0.5,-1.3) node[whitenode] (w2) [label=-90:$w_2$] {};

\draw (x) edge [post] node {} (w1);
\draw (x) edge [post] node {} (w2);
\draw (w1) edge [post] node {} (y);

\draw (4,0) node[set] (U) [label=right:$U'$] {};  
\draw (4,-1.5) node[set] (W) [label=right:$W$] {};

\draw (3.5,-0.2) node[whitenode] (x) [label=90:$x$] {}; 
\draw (4.5,-0.2) node[whitenode] (xy) [label=90:$xy$] {}; 

\draw (3.5,-1.3) node[whitenode] (w1) [label=-90:$w_1$] {}; 
\draw (4.5,-1.3) node[whitenode] (w2) [label=-90:$w_2$] {};

\draw (x) edge node {} (w1);
\draw (x) edge node {} (w2);
\draw (xy) edge node {} (w1);

\end{tikzpicture}
\caption{The digraph $D$ with vertex cover $U$ (left) and the corresponding graph $G_{D,U}$ (right) \label{fig:aux-graph}}
\end{figure}

For a given directed graph $D$ and a vertex cover $U$ in $D$ we build an undirected bipartite graph $B_{D,U}$ as follows (see Fig.~\ref{fig:aux-graph}).
Let $W=V(D)\setminus U$. Then,
\begin{eqnarray*}
 V(B) & = &U' \cup W\textrm{, where }U' = N^-(W)\cup (U\times U);\\[0.2cm]
 E(B) & = &\{\{xy,w\}\ :\ xy\in U\times U,\ w\in W, (x,w)\in E(D),\ (w,y)\in E(D)\} \cup \\
& & \{\{x,w\}\ :\ x\in U,\ w\in W,\ (x,w)\in E(D)\}.
\end{eqnarray*}

A \emph{crown decomposition} of an undirected graph $G$ is a partitioning of 
$V(G)$ into three parts $C$, $H$ and $R$, such that 
\begin{itemize}
\item $C$ is an independent set.
\item There are no edges between vertices of $C$ and $R$. That is, $H$ separates $C$ and $R$. 
\item $C$ can be partitioned into $C_m\cup C_u$ with $|C_m|=|H|$, such that $G[C_m\cup H]$ contains a perfect matching that matches each vertex of $C_m$ with a vertex of $H$.
\end{itemize}

Crown decompositions are used in multiple kernelization algorithms. In particular, the following lemma, which Gutin et al. attribute to Fellows et al.~\cite{FellowsHRST04}, shows that in certain situations a crown decomposition can be found efficiently.

\begin{lemma}[see~\cite{Gutin-quadratic}]\label{lem:crown-factory}
Suppose $G$ is an undirected graph on $n$ vertices, and suppose $I$ is an independent set in $G$ such that $|I|\geq \frac{2n}{3}$. Then $G$ admits a crown decomposition $(C=C_u\uplus C_m,H,R)$ with $C\subseteq I$, $H\subseteq V(G)\setminus I$ and $C_u\neq \emptyset$. Moreover, given $I$, the decomposition $(C=C_u\uplus C_m,H,R)$ can be found in $O(nm)$ time.
\end{lemma}

%

The main idea of Gutin et al. is to search for crowns in $B_{D,U}$ with $C\subseteq W$ and $C_u\neq \emptyset$. Such crowns can be conveniently reduced using the following reduction rule, whose correctness is proved in Lemma 4.4 of~\cite{Gutin-quadratic}. 

\setcounter{rulecnt}{0}
\rrule{r:crown}{Let $U$ be a vertex cover in $D$ and let $W=V(D)\setminus U$. 
Assume there is a crown decomposition $(C=C_m\cup C_u,H,R)$ in $B_{D,U}$ with $C\subseteq W$ and $C_u\neq \emptyset$.
Then remove $C_u$ from~$D$.}
\medskip

Our idea is to combine Rule~\ref{r:crown} with the knowledge that $D$ belongs to a graph class of bounded expansion $\mc{G}$, and hence Proposition~\ref{prop:bipartite-general} can be used to reason about the sparseness of the adjacency structure between $U$ and $W$. Let us introduce some notation. 
Consider a vertex cover $U$ and an independent set $W=V(D)\setminus U$ in $D$.
Let $W_s = \{w \in W\ :\ \deg_D(w) < 2\nabla_0(\mc{G})\}$, and let $W_b = W\setminus W_s$. 
Moreover, for $N\subseteq U$ with $|N|<2\nabla_0(\mc{G})$, let $W_N = \{w \in W_s\ :\ N(w) = N\}$. 
Let $\mc{N}(U)=\{N\subseteq U \ :\ |N|<2\nabla_0(\mc{G}), W_N \neq \emptyset\}$.      
Note that $|\mc{N}(U)| \le |W_s|$.

Our kernelization algorithm is as follows.

\begin{enumerate}
 \item If the algorithm from Lemma~\ref{lem:vertex-cover} returns an outbranching, answer YES and terminate; otherwise it returns a vertex cover $U$ of size at most $2k-2$. Let $W=V(D)\setminus U$.
 \item Construct the graph $B:=B_{D,U}$ and compute $W_s$, $\mc{N}(U)$, and nonempty sets $W_N$.
 \item If there is a set $N \in \mc{N}(U)$ such that $|W_N|>2|N_{B}(W_N)|$, then apply Lemma~\ref{lem:crown-factory} to graph $B[N_B[W_N]]$ with $I=W_N$. This gives us a crown decomposition $(C=C_u\uplus C_m,H,R)$ of $B[N_B[W_N]]$ with $C\subseteq W_N$, $H\subseteq N_{B}(W_N)$, and $C_u\neq \emptyset$. Observe that $(C=C_u\uplus C_m,H,R\cup (V(B)\setminus N_B[W_N]))$ is a crown decomposition of $B$.
 Apply Rule~\ref{r:crown} to this crown decomposition in order to remove $C_u$ from $D$, and restart the algorithm in the reduced graph.
 \item Otherwise, return $D$.      
\end{enumerate}

In case we have a prescribed root $r$ of the outbranching that we would like to preserve in the kernelization process, we can add it to the constructed vertex cover $U$, thus increasing its size up to at most $2k-1$. The reduction rules never remove any vertex of $U$.

Given this algorithm, we can restate and prove our main result for \probIOBshort.

\restateintroIOB*
\begin{proof}
 The correctness of our kernelization algorithm and a polynomial bound on its running time follows from Lemmas~\ref{lem:vertex-cover} and~\ref{lem:crown-factory}. 
 Note that the kernelization algorithm never decrements the budget $k$, so it suffices to show that it outputs an instance $(D,k)$ such that $|V(D)|=O(k)$.

 We can assume that the algorithm constructed a vertex cover $U$ of $D$ of size at most $2k-2$ ($2k-1$ if we want to preserve a prescribed root), because otherwise the algorithm would terminate and provide a positive answer. 
 Let $W = V(D) \setminus U$. 
 Then $V(D) = U \cup W_s \cup W_b$.
 By the first claim of Proposition~\ref{prop:bipartite-general} we get $|W_b| \le 2\nabla_0(G)|U| \le 4\nabla_0(\mc{G})k$.
 Hence it suffices to bound the size of $W_s$.
 Note that $W_s = \bigcup_{N \in \mc{N}(U)} W_N$.
 By the second claim of Proposition~\ref{prop:bipartite-general} we get $|\mc{N}(U)| \le (4^{\nabla_1(G)}+2\nabla_1(G)) |U| = O(4^{\nabla_1(\mc{G})}k)$.
 However, since Step 2 of the kernelization algorithm cannot be applied, for every $N \in \mc{N}(U)$ we have 
  $|W_N| \le 2|N_{B_{D,U}}(W_N)|$. However, by the construction of $B_{D,U}$ it is clear that $|N_{B_{D,U}}(W_N)| \le |N|^2 + |N|<4\nabla_0(\mc{G})^2+2\nabla_0(\mc{G})$, and hence $|W_N|<8\nabla_0(\mc{G})^2+4\nabla_0(\mc{G})$.
  It follows that $|W_s| = \sum_{N \in \mc{N}(U)} |W_N| = O(4^{\nabla_1(\mc{G})}\nabla_0(\mc{G})^2 k)$, and hence $|V(D)|=|U|+|W_s|+|W_b|=O(4^{\nabla_1(\mc{G})}\nabla_0(\mc{G})^2 k)$.
  This finishes the proof.
\end{proof}

Let us remark that in the proof of Theorem~\ref{thm:mainIOB} we used only the boundedness of $\nabla_0(\mc{G})$ and $\nabla_1(\mc{G})$, so our algorithm works as well in any graph class where only these two grads are finite constants. Also, the kernelization algorithm has polynomial running time, where the degree of the polynomial is a constant independent of $\mc{G}$.

\section{Subexponential algorithms}\label{sec:subexp}

Theorems~\ref{thm:mainLOB} and~\ref{thm:mainIOB} enable us to design subexponential parameterized algorithms for \probRMLOshort and \probIOBshort on $H$-minor-free graphs using the standard approach via treewidth. To this end, we compose two facts: First, for a fixed forbidden minor $H$, every $H$-minor-free graph on $n$ vertices has treewidth at most $O(\sqrt{n})$~\cite{Grohe03}. Second, both \probRMLO and \probIOB can be solved in time $2^{O(t)}\cdot n^{O(1)}$ on $n$-vertex graphs given together with their tree decompositions of width at most $t$, as explained next.

For the latter ingredient, a standard approach to dynamic programming on tree decompositions would yield algorithms with running time $2^{O(t\log t)}\cdot n^{O(1)}$, since we consider all possible partitions of a bag in the states of the dynamic programming table. However, both problems are amenable to recently developed new techniques for constructing dynamic programming algorithms with running time $2^{O(t)}\cdot n^{O(1)}$. An application of the Cut\&Count technique~\cite{CyganNPPRW11} immediately yields randomized algorithms with such a running time for both these problems. Actually, the existence of such algorithms also follows from expressibility in the logical formalism ECML+C proposed by the third author~\cite{Pilipczuk11,abs-1104-3057}, which provides a meta-result on applicability of Cut\&Count; The full paper~\cite{abs-1104-3057}, Appendix D, contains a formula for the problem of finding an outbranching with exactly $k$ leaves, which can be trivially adjusted to express both \probRMLO and \probIOB. The Cut\&Count technique has been recently derandomized by Bodlaender et al.~\cite{BodlaenderCKN13}, who proposed the so-called {\em{rank based approach}} that yields deterministic $2^{O(t)}\cdot n^{O(1)}$-time algorithms for many problems amenable to Cut\&Count. It is a simple exercise to see that using this technique one can also design such algorithms for \probRMLO and \probIOB. Thus, we have the following proposition.

\begin{proposition}\label{prop:tw-DP}
\probRMLO and \probIOB can be solved in deterministic time $2^{O(t)}\cdot n^{O(1)}$ on an $n$-vertex graph given together with its tree decomposition of width $t$.
\end{proposition}

Gathering all the tools, we obtain the subexponential algorithms promised in Section~\ref{sec:intro}.

\restateintrosubexp*
\begin{proof}
Let $(D,k)$ be the input instance of \probRMLOshort or \probIOBshort, where $D$ is $H$-minor free. First, we apply the kernelization algorithm of Theorem~\ref{thm:mainLOB} or~\ref{thm:mainIOB} (depending on the problem) to reduce the size of the instance to $O(k)$; note that the application of neither of these algorithms can increase the parameter. Having the reduced instance $(D',k')$ in hand, where $k'\leq k$, $D'$ is $H$-minor-free, and $|V(D')|=O(k)$, we infer that the treewidth of $D'$ is in $O(\sqrt{k})$. Hence we apply any constant-factor approximation algorithm for treewidth, e.g.~\cite{BodlaenderDDFLP13}, to compute a tree decomposition of $D'$ of width $O(\sqrt{k})$ in time $2^{O(\sqrt{k})}$. We conclude by applying the appropriate algorithm of Proposition~\ref{prop:tw-DP}; this application also takes time $2^{O(\sqrt{k})}$.
\end{proof}

We remark that the running time of the algorithms given by Theorem~\ref{thm:subexp} is essentially optimal under the assumption of the Exponential Time Hypothesis (ETH), even already on planar directed graphs. More precisely, from the known NP-hardness reductions it follows that the existence of an algorithm for \probRMLOshort or \probIOBshort working in time $2^{o(\sqrt{N})}$ on a planar directed graph with $N$ vertices would contradict ETH. For completeness, we sketch now how this conclusion can be derived.

\begin{theorem}
Unless ETH fails, there is no algorithm solving \probRMLOshort or \probIOBshort that achieves running time $2^{o(\sqrt{N})}$ on planar directed graphs with $N$ vertices.
\end{theorem}
\begin{proof}
For \probIOBshort the statement follows easily from the known fact that the existence of such an algorithm for {\sc{Planar Hamiltonian Cycle}} would contradict ETH, see e.g.~\cite{LokshtanovMS11}. First, {\sc{Planar Hamiltonian Cycle}} can be Turing-reduced to {\sc{Planar Hamiltonian Path}} by guessing an edge used in the solution and replacing it with two pendant vertices attached to its endpoints. {\sc{Planar Hamiltonian Path}} can be now reduced to the variant of \probIOBshort on planar digraphs, where the root is not specified: simply replace every undirected edge by two directed edges with opposite orientations, and ask for an outbranching with at least $N-1$ internal vertices. The variant with unspecified root is easily Turing-reducible to the one with specified root by simply guessing the root. Note that all the aforementioned reductions increase the instance size by at most a constant factor, and hence the statement for \probIOBshort follows.

We turn our attention to \probRMLOshort. First, it is known that the {\sc{Planar Vertex Cover}} problem does not admit an algorithm with running time $2^{o(\sqrt{N})}$ on planar graphs with $N$ vertices, see e.g.~\cite{grohe:book,LokshtanovMS11}. Garey and Johnson~\cite{GareyJ77} proposed a reduction from {\sc{Planar Vertex Cover}} to {\sc{Planar Connected Vertex Cover}} that increases the number of vertices of the graph only by a constant multiplicative factor. This proves that also for {\sc{Planar Connected Vertex Cover}} an algorithm with running time $2^{o(\sqrt{N})}$ can be excluded under ETH. 

We further reduce {\sc{Planar Connected Vertex Cover}} to {\sc{Planar Connected Dominating Set}} using the following transformation: subdivide every edge of the graph, and add a pendant to every introduced subdividing vertex. It can be easily shown that a planar graph $G$ has a connected vertex cover of size $k$ if and only if the planar graph $G'$ obtained in this transformation has a connected dominating set of size $k+|E(G)|$. Hence, under ETH there is no $2^{o(\sqrt{N})}$-time algorithm for {\sc{Planar Connected Dominating Set}}.

Now, we use the known fact that the {\sc{Connected Dominating Set}} is dual to the {\sc{Max Leaf}} problem --- the problem of finding a spanning tree of an undirected graph with the maximum possible number of leaves. More precisely, a graph $G$ has a connected dominating set of size at most $k$ if and only if it admits a spanning tree with at least $|V(G)|-k$ leaves; cf.~\cite{Douglas92}. Hence, under ETH there is no $2^{o(\sqrt{N})}$-time algorithm for {\sc{Planar Max Leaf}}.

Finally, {\sc{Planar Max Leaf}} can be reduced to the variant of \probRMLOshort on planar graphs where the root is not specified by just replacing every undirected edge by two directed edges with opposite orientations. Again, the variant with unspecified root is easily Turing-reducible to the one with specified root by simply guessing the root. This proves the statement for \probRMLOshort.
\end{proof}

\section{Concluding remarks}
\label{sec:conclude}

In this paper we have shown linear kernels for both \probRMLO and \probIOB on sparse graph classes: $H$-minor-free and of bounded expansion, respectively. We believe that our work is another good example of how abstract properties derived from the sparsity of the considered graph class, in particular the ones expressed in Proposition~\ref{prop:bipartite-general}, can be used in the kernelization setting for a clean treatment of graph classes with excluded minors, without the need of invoking the decomposition theorem of Robertson and Seymour. Other examples of this approach include~\cite{DrangeDFKLPPRSVS14,BndExpKernels13}, and we hope that even more will appear in future.

In the light of our results, the question about the existence of linear kernels for \probRMLO and \probIOB on general graphs becomes even more tantalizing. We do not intend to take a stance about the actual answer, but after investigating both problems for some time we believe that in both cases a conceptual breakthrough is needed to make an improvement.

\section*{Acknowledgments}

The authors are very grateful to Marcin Pilipczuk for reading the manuscript carefully and providing useful comments.

\bibliographystyle{abbrv}

\end{document}